\documentclass{sig-alternate}
\PassOptionsToPackage{usenames,dvipsnames}{xcolor}
\usepackage{wasysym}

\newcommand\M{\mathsf{B}}

\newcommand{\nextMode}{\textsc{nextMode}}
\newcommand{\getRateFromEnv}{\textsc{senseCurrentRate}}
\newcommand{\updateProjection}{\textsc{updateProjection}}
\newcommand{\reduceComp}{\textsc{reduceComp}}

\usepackage{tikz}
\usepackage{pgf}
\usepackage{pgflibraryarrows}
\usepackage{pgffor}
\usepackage{pgflibrarysnakes}
\usepackage{todonotes}
\usetikzlibrary{fit} 
\usetikzlibrary{backgrounds} 
\usetikzlibrary{patterns}
\usetikzlibrary{positioning}
\tikzstyle{background}=[rectangle,fill=gray!10, inner sep=0.1cm, rounded corners=0mm]
\usepgflibrary{shapes}
\usetikzlibrary{snakes,automata}
\usetikzlibrary{shadows}
\usepackage[usenames,dvipsnames]{xcolor}
\tikzstyle{background}=[rectangle,fill=gray!10, inner sep=0.1cm, rounded corners=0mm]
\tikzstyle{loc}=[draw,rectangle,minimum size=1.4em,inner sep=0em]
\tikzstyle{trans}=[-latex, rounded corners]
\tikzstyle{trans2}=[-latex, dashed, rounded corners]
\tikzstyle{player2}=[draw,dashed,minimum size=5mm]
\makeatletter
\newif\if@restonecol
\makeatother

\usepackage[ruled,vlined,linesnumbered,lined]{algorithm2e}

\usepackage{xspace}

\definecolor{lightgray}{gray}{0.9}

\DeclareMathOperator{\interior}{int}
\newcommand{\norm}[1]{\|#1\|}

\mathchardef\breakingcomma\mathcode`\,
{\catcode`;=\active
  \gdef;{\breakingcomma\discretionary{}{}{}}
}

\newcommand{\point}[1]{{\overline{#1}}}

\newcommand{\px}{\point{x}}
\newcommand{\py}{\point{y}}

\newcommand{\vv}{\vec{v}}

\newcommand{\vb}{\vec{b}}

\newcommand{\vt}{\vec{t}}
\newcommand{\vr}{\vec{r}}

\newcommand{\Ext}{\text{\it Ext}}
\newcommand{\Vtx}[1]{\mathsf{vert}(#1)}

\newcommand{\ball}[2]{B_{#1}(#2)}

\newcommand{\vzero}{\vec{0}}

\newcommand{\set}[1]{\left\{ #1 \right\}}

\newcommand{\seq}[1]{\langle #1 \rangle}
\newcommand{\Rplus}{{\mathbb R}_{\geq 0}}
\newcommand{\Nat}{\mathbb N}

\newcommand{\Real}{\mathbb R}

\newcommand{\Aa}{\mathcal{A}}

\newcommand{\Hh}{\mathcal{H}}

\newcommand{\Ff}{\mathcal{F}}

\newcommand{\Rr}{\mathcal{R}}

\newcommand{\Tt}{\mathcal{T}}

\DeclareMathOperator*{\argmin}{arg\,min}

\newcommand{\zzero}{\mathbf{0}}

\newcommand{\sem}[1]{ [ \! [ {#1}  ]  \! ]} 

\usepackage{color}

\newcommand{\Ashu}[1]{}
\newcommand{\Umang}[1]{}
\newcommand{\Umangm}[1]{}

\newcommand{\FRUNS}{\text{\it FRuns}}
\newcommand{\RUNS}{\text{\it Runs}}
\newcommand{\RUN}{\text{\it Run}}

\newtheorem{theorem}{Theorem}
\newtheorem{example}{Example}
\newtheorem{definition}{Definition}

\newtheorem{lemma}[theorem]{Lemma}
\newtheorem{proposition}[theorem]{Proposition}
\newtheorem{corollary}[theorem]{Corollary}

\newcommand{\RCH}{\mathcal{W}_\mathrm{Reach}}

\newcommand{\cmms}{\textsf{CMS}}
\newcommand{\bmms}{\textsf{BMS}}
\newcommand{\mmms}{\textsf{MMS}}

\usepackage{hyperref}

\begin{document}

\title{Bounded-Rate Multi-Mode Systems Based Motion Planning}

\numberofauthors{4} \author{ 
  \alignauthor \scalebox{0.8}{Devendra Bhave}\\
  \affaddr{IIT Bombay}
\email{\scalebox{0.8}{devendra@cse.iitb.ac.in}}
  \alignauthor \scalebox{0.8}{Sagar Jha}\\
  \affaddr{IIT Bombay}
  \email{\scalebox{0.8}{sagarjha@cse.iitb.ac.in}}
   \alignauthor \scalebox{0.8}{Shankara Narayanan Krishna}\\
  \affaddr{IIT Bombay}
  \email{\scalebox{0.8}{krishnas@cse.iitb.ac.in}}
  \and
     \alignauthor \scalebox{0.8}{Sven Schewe}\\
 \affaddr{University of Liverpool}
 \email{\scalebox{0.8}{sven.schewe@liverpool.ac.uk}}
     \alignauthor \scalebox{0.8}{Ashutosh Trivedi}\\
     \affaddr{IIT Bombay}
 \email{\scalebox{0.8}{trivedi@cse.iitb.ac.in}}
}

\date{\today}

\maketitle

\begin{abstract}
Bounded-rate multi-mode systems are hybrid systems that can switch among a
finite set of modes. Its dynamics is specified by a finite number of
real-valued variables with mode-dependent rates that can vary within given
bounded sets. Given an arbitrary piecewise linear trajectory, we study the
problem of following the trajectory with arbitrary precision, using motion
primitives given as bounded-rate multi-mode systems. We give an algorithm to
solve the problem and show that the problem is co-NP complete. We further prove
that the problem can be solved in polynomial time for multi-mode systems with
fixed dimension. We study the problem with dwell-time requirement and show the
decidability of the problem under certain positivity restriction on the rate
vectors. Finally, we show that introducing structure to the multi-mode systems leads to undecidability,
even when using only a single clock variable.
\end{abstract} 




\section{Introduction}
\label{sec:introduction}
Hybrid automata~\cite{ACHH92} are a natural and expressive formalism to model systems that
exhibit both discrete and continuous behavior. 
Intuitively, hybrid automata extend the discrete system modeling framework of
extended finite state machines with continuous variables modeled along
continuous dynamical systems such that the flow of continuous variables in each
state is modeled as a system of first-order ordinary differential equations.
Discrete jumps in the values of the variables are  modeled via resets
on the transitions of the automata. 
However, the applications of hybrid automata in analyzing cyber-physical systems
have been rather limited due to undecidability~\cite{HKPV98} of simple
verification problems such as reachability. 
This drawback of hybrid automata has fueled the investigation of the
so-called compositional methodology~\cite{deAH01,LeNP12} to design complex
system by sequentially composing well-understood lower-level components.
This methodology has, for example, been used in the context of the \emph{motion
planning} problem for mobile robots, where the task is to move a robot along a
pre-specified trajectory with arbitrary precision by sequentially composing a
set of well-studied simple motion primitives, such as ``move left'', ``move right''
and  ``go straight''.  
In this paper, we investigate the motion planning problem for systems, whose motion
primitives are given as constant-rate vectors with uncertainties.

We consider bounded-rate multi-mode systems~\cite{AFMT13} that can be considered
as \emph{constant-rate multi-mode systems~\cite{ATW12} with uncertainties}.
These systems consist of a finite set of continuous variables, whose dynamics is
given by mode-dependent constant-rates that can vary within given bounded sets. 
In such systems, the dynamics of the system can be viewed as a two-player
game between a scheduler and the environment.
In each step, the scheduler chooses a mode and time duration and the
environment chooses a rate vector for that mode from the given bounded set.
The system evolves with that rate for the chosen time. The game continues in
this fashion from the resulting state. 
Alur, Trivedi, and Wojtczak~\cite{ATW12} considered constant-rate multi-mode
systems and showed that the reachability problem---deciding the reachability of
a specified state while staying in a given safety set---and the schedulability
problem---deciding the existence of a non-Zeno control so that the system
always stays in a given bounded and convex safety set---for this class of systems
can be solved in polynomial time.  
Alur et al.~\cite{AFMT13} showed that the existence of robust control for the
schedulability problem for bounded-rate multi-mode systems is, although
intractable (co-NP-complete), decidable. 
However, they left the decidability of the robust reachability problem for this
class of systems open. 

The robust reachability problem for bounded-rate multi-mode system is defined as
follows: given a bounded-rate multi-mode system, a starting state, and a target
state, decide whether it is possible to reach the target state from the starting
state with arbitrary precision. 
The key result of this paper is the decidability of the robust reachability
problem for bounded-rate multi-mode systems. 
We show that the problem is co-NP complete. 
Moreover, we show that it is fixed parameter tractable, i.e., if the number of
dimensions is fixed, then the robust reachability problem can be solved in
polynomial time.

Our existence proofs are constructive: in case of a positive answer, we can also
give a dynamic schedule that, given a tolerance level $\varepsilon {>} 0$,
guarantees reachability of an open ball of $\varepsilon$ radius around the
target state in finitely many steps.  
It is then simple to extend these results to different path planning problems.
We discuss the extension of the robust reachability problem to motion planning, and
exploit our results to provide an alternative and simpler proof for the decidability of the
robust schedulability problem. We also show that this problem can be solved in
polynomial time for systems with fixed dimension, improving the
result~\cite{AFMT13} where authors only give a polynomial algorithm to decide
$2$-dimensional systems. 
We notice that these results can be combined to \emph{stable reachability},
where the goal is to first reach an $\varepsilon$ ball around a target, and then
stay in this ball for ever. 

\begin{figure}[t]
\begin{center}
\scalebox{1}{
\begin{tikzpicture}[->,>=stealth',shorten >=1pt,auto,node distance=0.5cm,
  semithick]
    \draw[-,dotted] (-3, 0) -- (3, 0);
   \draw[-,dotted] (0, -2) -- (0, 3);

   \draw node (r1) at (-1,-1){};
   \draw node (r2) at (1,-1){};
   \draw node (r3) at (-1.2,1.6){};
   \draw node (r4) at (1.2,1.6){};

   \draw[fill=black] (r1) circle(0.05);
   \draw[fill=black] (r2) circle(0.05);
   \draw[fill=black] (r3) circle(0.05);
   \draw[fill=black] (r4) circle(0.05);

   \node [below of=r1] {(-1, -1)};
   \node [below of=r2] {(1, -1)};
   \node [above of=r3] {$(\frac{-6}{5}, \frac{8}{5})$};
   \node [above of=r4] {$(\frac{6}{5}, \frac{8}{5})$};
   \node [above left of=r1, node distance=0.7cm] {$m_1 = \{\vr_1\}$};
   \node [above right of=r2, node distance=0.7cm] {$m_2 = \{\vr_2\}$};
   \node [below of=r3] {$\vr_3$};
   \node [below of=r4] {$\vr_4$};
   \node at (0, 1.8) {$m_3$};

   \draw[->, color=red] (0,0) -- (-1,-1);
   \draw[->, color=green] (0,0) -- (1,-1);
   \draw[->, color=blue] (0,0) -- (-1.2,1.6);
   \draw[->, color=blue] (0,0) -- (1.2,1.6);
   \draw[line width=2pt, -, color=blue] (-1.2,1.6) -- (1.2,1.6);
   
   \fill[blue!20,rounded corners, fill opacity=0.2] (-0.3, -0.3) rectangle (0.3,
   3.7);

   \draw[thick] node (starto) at (0, 0) {$$};
   \draw[fill=black!40] (starto) circle(0.05);
   \node [left of=starto] {$\px_0$};

   \draw[thick] node (end) at (0, 3) {$$};
   \draw[fill=black!40] (end) circle(0.05);
   \node [left of=end] {$\px_t$};
   \draw[dashed] (end) circle(0.25);
   
\end{tikzpicture}}
\caption{Bounded-rate multi-mode system with three modes and two variables.} 
\label{fig:intro-example}
\end{center}
\end{figure}
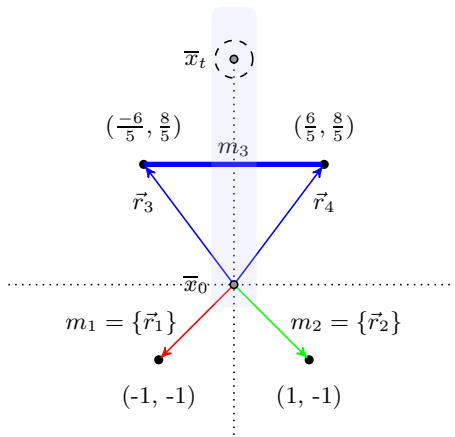

\begin{example}
An example of a bounded-rate multi-mode system with two variables, say $x$ and
$y$, and three modes $m_1$, $m_2$, and $m_3$, is given in the
Figure~\ref{fig:intro-example}.         
Modes $m_1$ and $m_2$ are precise, while mode $m_3$ is uncertain, and
environment can give any rate vector that is a convex combination of rate
vectors $\vr_3$ and $\vr_4$.
The safety set is given as the blue rectangle. 
The reachability problem here is to decide whether, for every
$\varepsilon > 0$,  scheduler has a sequence of time delays and choice of modes
such that no matter what rate is given by the environment the system reaches a
state in $\varepsilon$-neighborhood of $\px_t$. 
The schedulability problem asks whether the scheduler has an infinite
non-Zeno sequence of choices of modes and time delays such that the system
always stays within the safety set, while stable reachability problem
asks for  a strategy to first reach an $\varepsilon$-neighborhood of $\px_t$ and
then to stay in that neighborhood using a non-Zeno strategy.
\end{example}

We also consider the reachability problem with minimum dwell-time requirement
and show that in the absence of the safety set the problem is undecidable for
arbitrary bounded-rate multi-mode systems, but turns out to be decidable for
systems with non-negative rates. 
We also study the problem of the existence of discrete control where scheduler
is required to choose modes at times multiple of a given sampling rate. 
We show that the reachability problem is EXPTIME-complete for this class of
schedulers.  
Finally, we show that adding very simple structure to bounded-rate multi-mode
systems  by introducing clock variables (variables with precise uniform rates in
each mode)---that appear as guards on the transitions and can be reset on
the discrete transitions---leads to undecidability of the robust reachability
problem. 

Our algorithm can be combined with algorithms to explore non-convex
high-dimensional spaces, such as rapidly exploring random tree (RRT)
algorithm~\cite{rrt-plan}, to yield robust control for such systems.  
Intuitively, RRT algorithm can return a path from the source to the destination
by random exploration of the state space, which can be robustly followed by
repeated applications of our algorithm in context of systems modeled as
bounded-rate multi-mode systems.

For a review of related work on constant-rate multi-mode systems we refer the
reader to~\cite{AFMT13,ATW12}.
Le Ny and Pappas~\cite{LeNP12} initiated work on the sequential composition of
robust controller specifications. 
In this light, our results can be understood as an effort to analyze complexity of this problem
for the system of relatively simple dynamics. 
There is a huge body of work on path-following and trajectory tracking of
autonomous robots under uncertainty. 
For a detailed survey we refer the reader to~\cite{AH07}.
There is a vast literature on decidable subclasses of hybrid
automata~\cite{ACHH92,BBM98}. 
Most notable among these classes are initialized rectangular hybrid
automata~\cite{HKPV98}, two-dimensional piecewise-constant derivative
systems~\cite{AMP95}, and 
timed automata~\cite{alurDill94}.

The paper is organized as follows. 
We begin by formal definition of the problem in the next section, followed by
the proof of our key result in Section~\ref{sec:key_result}.
We present some applications of our main algorithm to solve schedulability,
stable reachability, and path following problems in Section~\ref{apps}. 
In Section~\ref{sec:dwell} we present results regarding bounded-rate multi-mode
systems with discrete scheduler and dwell-time requirements.
We conclude the paper by discussing results on generalized model in
Section~\ref{sec:generalized}.

\section{Robust Reachability Problem}
Prior to formally introducing the robust reachability problem for multi-mode
systems, we set the notation used in the rest of the paper and recall some
standard results. 

\subsection{Preliminaries}
{\bf Points and Vectors.} 
Let $\Real$ be the set of real numbers.
We represent the states in our system as points in $\Real^n$ that is equipped
with the standard \emph{Euclidean norm} $\norm{\cdot}$.
We denote points in this state space by $\px, \py$,  vectors by $\vr, \vv$, and 
the $i$-th coordinate of point $\px$ and vector $\vr$ by $\px(i)$ and $\vr(i)$,
respectively. 
We write $\vzero$ for a vector with all its coordinates equal to $0$; 
its dimension is often clear from the context.
The distance $\norm{\px, \py}$ between points $\px$ and $\py$ is defined as
$\norm{\px - \py}$. 

{\bf Boundedness and Interior.} 
We denote an {\em open ball} of radius $d \in \Rplus$ centered at $\px$ as
$\ball{d}{\px} {=} \set{\py {\in} \Real^n \::\: \norm{\px,\py} < d}$.
We denote a closed ball of radius $d \in \Rplus$ centered at $\px$ as
$\overline{\ball{d}{\px}}$. 
We say that a set $S \subseteq \Real^n$ is {\em bounded} if there exists
$d \in \Rplus$ such that, for all $\px, \py \in S$, we have
$\norm{\px,\py} \leq d$.
The {\em interior} of a set $S$, $\interior(S)$, is the set of all points
$\px \in S$, for which there exists $d > 0$ s.t. $\ball{d}{\px} \subseteq S$.

{\bf Convexity.} A point $\px$ is a \emph{convex
  combination} of a finite set of points $X = \set{\px_1, \px_2, \ldots, \px_k}$ if
there are $\lambda_1, \lambda_2, \ldots, \lambda_k \in [0, 1]$ such that
$\sum_{i=1}^{k} \lambda_i = 1$ and $\px = \sum_{i=1}^k \lambda_i \cdot \px_i$.
The \emph{convex hull} of $X$ is the set of all points
that are convex combinations of points in $X$.
We say that $S \subseteq \Real^n$ is {\em convex} iff, for all
$\px, \py \in S$ and all $\lambda \in [0,1]$, we have
$\lambda \px + (1-\lambda) \py \in S$ and moreover,
$S$ is a {\em convex polytope} if it is bounded and there exists $k \in \Nat$, a matrix $A$ of 
size $k \times n$ and a vector $\vb \in \Real^k$ such that $\px \in S$ iff
$A\px \leq \vb$. 

%
A point $\px$ is a \emph{vertex} of a convex polytope $P$ if it is
not a convex combination of two distinct (other than $\px$) points in $P$. 
For a convex polytope $P$ we write $\Vtx{P}$ for the finite set of points that
correspond to the vertices of $P$.  
Each point in $P$ can be written as a convex combination of the points in
$\Vtx{P}$. In other words, $P$ is the \emph{convex hull} of $\Vtx{P}$.

\subsection{Multi-Mode Systems}
A multi-mode system is a hybrid system, or rather a switched system, equipped
with finitely many \emph{modes} and finitely many real-valued \emph{variables}. 
A configuration is described by the values of the variables. These values change as
time elapses at the rates determined by the modes being used. 
The choice of the rates is nondeterministic, which introduces
a notion of adversarial behavior.
\begin{definition}[Multi-Mode Systems]
  \label{def:BMMS}
  A multi-mode system is a tuple $\Hh = (M, n, \Rr)$ where:
    $M$ is the finite nonempty set of \emph{modes}, 
    $n$ is the number of continuous variables, and 
    $\Rr : M \to 2^{\Real^n}$ is the rate-set function that, for each mode $m \in 
    M$, gives a set of vectors.
We often write  $\vr \in m$ for $\vr\in \Rr(m)$ when $\Rr$ is clear
from the context. 
\end{definition}

A finite \emph{run} of a multi-mode system $\Hh$ is a finite sequence of states,
timed moves, and rate vector choices  
$\varrho = \seq{\px_0, (m_1, t_1), \vr_1, \px_1,  \ldots, (m_k, t_k), \vr_k,
  \px_k}$ 
s.t., for all $1 \leq i \leq k$, we have $\vr_i \in \Rr(m_i)$ and   
$\px_i = \px_{i-1} + t_i \cdot \vr_i$.
For such a run $\varrho$ we say that $\px_0$ is the \emph{starting state}, while
$\px_k$ is its \emph{last state}.
An infinite run is defined in a similar manner. 
We write $\RUNS$ and $\FRUNS$ for the set of infinite and finite runs of $\Hh$,
and $\RUNS(\px)$ and $\FRUNS(\px)$  for the set of infinite and finite runs of $\Hh$
that start from $\px$.  

An infinite run $\seq{\px_0, (m_1, t_1), \vr_1, \px_1, (m_2, t_2),
  \vr_2,  \ldots}$ is \emph{Zeno} if $\sum_{i=1}^{\infty} t_i < \infty$.
Given a set $S \subseteq \Real^n$ of safe states, we say that a  run 
$\seq{\px_0, (m_1, t_1), \vr_1, \px_1, (m_2, t_2), \vr_2,  \ldots, (m_k,
t_k), \vr_k, \px_k}$ is $S$-safe if $\px_i \in S$ for all $0 {\leq} i {\leq} k$;
and for all $0 {\leq} i {<} k$ we have that $\px_i +
t \cdot \vr_{i+1} \in S$ for all $t \in [0, t_{i+1}]$, assuming $t_0 = 0$. 
Notice that, if $S$ is a convex set and $\px_i \in S$ for all $i\ge 0$, then this holds iff $\px_i \in S$ for
all $0 {\leq} i {\leq} k$.
Sometimes we simply call a run safe when the safety set is clear from the context. 

We formally give the semantics of a multi-mode system $\Hh$ as a
turn-based two-player game between two players, \emph{scheduler} and
\emph{environment}, who choose their moves to construct a run of the system.   
The system starts in a given starting state $\px_0 \in \Real^n$. At each turn, the
scheduler chooses a timed move, a pair $(m, t) \in M \times \Real_{>0}$ consisting
of a mode and a time duration, and the environment chooses a rate vector 
$\vr \in m$ and as a result the system changes its state from $\px_0$ to
the state $\px_1 = \px_0 + t \cdot \vr$ in $t$ time units following the linear
trajectory according to the rate vector $\vr$.
From the next state, $\px_1$, the scheduler again chooses a timed move and the
environment an allowable rate vector, and the game continues forever in this
fashion.
The focus of this paper is on robust \emph{reachability} problem where, given a
starting state $\px_0$, a target vertex $\px_t$, a bounded and convex safety set
$S$ and tolerance $\varepsilon > 0$, the goal of the scheduler is to visit a
state in an open ball of radius $\varepsilon$ centered at $\px_t$ via an $S$-safe run.   
The goal of the environment is the opposite.

Given a bounded and convex safety set $S$ and tolerance $\varepsilon {>} 0$,  we
define the \emph{robust reachability objective} $\RCH^S(\px_t, \varepsilon)$ as the set of
infinite runs of $\Hh$ that visit a state in $\ball{\varepsilon}{\px_t}$.  
In a reachability game the winning objective of the scheduler is to make sure
that the constructed run of a system belongs to $\RCH^S(\px_t, \varepsilon)$, while the
goal of the environment is the opposite. 
The choice selection mechanism of the players is typically defined as
strategies. 
A \emph{strategy} $\sigma$ of the scheduler is function  $\sigma{:} \FRUNS {\to} M
{\times} \Rplus$ that gives a timed move for every history of the game.
A strategy $\pi$ of the environment is a function $\pi: \FRUNS \times (M \times
\Rplus) \to \Real^n$ that chooses an allowable rate for a given history of the
game and choice of the scheduler.
 We write $\Sigma$ and $\Pi$ for the set of strategies of the scheduler and the
 environment, respectively.

Given a starting state $\px_0$ and a strategy pair $(\sigma, \pi) \in \Sigma
\times \Pi$ we define the unique run $\RUN(\px_0, \sigma, \pi)$ 
starting from $\px_0$ as 
\[
\RUN(\px_0, \sigma, \pi)=\seq{\px_0, (m_1, t_1), \vr_1, \px_1, (m_2, t_2),
  \vr_2,  \ldots}
\]
where, for all $i{\ge} 1$, $(m_i, t_i) = \sigma(\seq{\px_0, (m_1, t_1), \vr_1, \px_1, \ldots,\px_{i-1}})$ and 
$\vr_i=\pi(\seq{\px_0, (m_1, t_1), \vr_1, \px_1, \ldots,\px_{i-1},m_i, t_i})$ and
$\px_i= \px_{i-1} + t_i\cdot \vr_i$.
The scheduler wins the game if there is a $\sigma \in \Sigma$ such that, for all $\pi \in \Pi$,
we get $\RUN(\px_0, \sigma, \pi) \in \RCH^S(\px_t, \varepsilon)$.   
Such a strategy $\sigma$ is \emph{winning}.  
Similarly, the environment wins the game if there is $\pi \in \Pi$ such that
for all $\sigma \in \Sigma$ we have $\RUN(\px_0, \sigma, \pi) \not \in
\RCH^S(\px_t, \varepsilon)$.
Again, $\pi$ is called \emph{winning} in this case.
If a winning strategy for scheduler exists, we say that the state $\px_t$ is
$\varepsilon$-reachable from the state $\px_0$ for given safety set $S$
and tolerance $\varepsilon$. 
We also say that the state $\px_t$ is {\em robustly reachable} from $\px_0$ if
it is $\varepsilon$-reachable for all $\varepsilon > 0$. 
The following is the main algorithmic problem studied in this paper.
\begin{definition}[Robust Reachability]
  \label{def:schedulability}
  Given a multi-mode system $\Hh$, a convex safety set $S$,  a starting
  state $\px_0 \in \interior(S)$, and a target state $\px_t \in \interior(S)$,
  decide whether $\px_t$ is robustly reachable from $\px_0$. 
\end{definition}

To algorithmically decide the robust reachability problem, we need to restrict
the range of $\Rr$ and the domain of the safety set $S$ in a robust reachability
game on a multi-mode system. 
The most general model that we consider is the bounded-rate multi-mode systems
(\bmms{}).

\begin{definition}[Bounded-Rate Systems] 
A bounded-rate multi-mode system (\bmms{})is multi-mode system $\Hh = (M, n,
\Rr)$ such that $\Rr(m)$ is a convex polytope for every $m \in M$. 
We also assume that the safety set $S$ is specified as a convex polytope.
\end{definition}

For every mode $m_i \in M$ of a \bmms{} we assume an arbitrary but fixed
ordering on the vertices of $\Rr(m)$. 
By exploiting the notations slightly, it allows us to write $\Rr(m_i)(j)$ for
the rate vector corresponding to $j$-th vertex of mode $m_i$. 
When there is no confusion, we also write $\Rr(i)(j)$ for $\Rr(m_i)(j)$.

In our proofs we often refer to another variant of multi-mode systems, in which
there are only a fixed number of different rates in each mode (i.e., $\Rr(m)$ is
finite for all $m\in M$).   
We call such a multi-mode system \emph{multi-rate multi-mode systems}
(\mmms{}). 
Finally, a special form of \mmms{} are  \emph{constant-rate multi-mode systems}
(\cmms{})~\cite{ATW12}, in which $\Rr(m)$ is a singleton for all $m\in M$.
We sometimes use $\Rr(m)$ to refer to the unique element of the set $\Rr(m)$ in
a \cmms{}.
The concepts related to the robust reachability games for \bmms{} and \mmms{} are
already defined for multi-mode systems. 
Similar concepts also hold for \cmms{} but with no real choice for the
environment. 
Examples of \cmms{}, \bmms{}, and \mmms{} are shown in
Figure~\ref{fig:example1}. 

We say that a \cmms{} $H = (M, n, R)$ is an instance of a multi-mode system 
$\Hh = (M, n, \Rr)$ if for every $m \in M$ we have that $R(m) \in \Rr(m)$.
For example, the \cmms{} shown in Figure~\ref{fig:example1}.(a) is an instance
of \bmms{} in Figure~\ref{fig:example1}.(b).
We denote the set of instances of a multi-mode system $\Hh$ by $\sem{\Hh}$. 
Notice that for a \bmms{} $\Hh$, the set $\sem{\Hh}$ of its instances is
uncountable (unless the \bmms {} is a \cmms {}), while for an \mmms{} $\Hh$ the set $\sem{\Hh}$ is finite,
and exponential in the size of $\Hh$. 
We say that an \mmms{} $(M, n, \Rr')$ is the \emph{extreme-rate} \mmms{} of a
\bmms{} $(M, n, \Rr)$ if $\Rr'(m) = \Vtx{\Rr(m)}$. 
The \mmms{} in Figure~\ref{fig:example1}.(c) is the extreme-rate \mmms{} for the
\bmms{} in Figure~\ref{fig:example1}.(b) 
We write $\Ext(\Hh)$ for the extreme-rate \mmms{} of the \bmms{} $\Hh$.
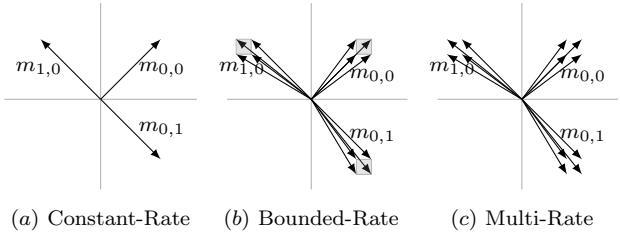
\begin{figure}[t]
  \centering
  {\small
  \begin{tikzpicture}[node distance=2cm,scale=0.4]
    \tikzstyle{lines}=[draw=black!30,rounded corners]
    \tikzstyle{vectors}=[-latex, rounded corners]
    \tikzstyle{rvectors}=[-latex,very thick, rounded corners]

    \draw[lines] (-3.2,0)--(3.2,0);
    \draw[lines] (0, 3.2)--(0,-3);
    \draw[vectors] (0, 0) --node[right]{$m_{0,0}$} (2, 2) node[left]{$$};
    \draw[vectors] (0, 0) --node[right]{$m_{0,1}$} (2, -2)node[right]{$$};
    \draw[vectors] (0, 0) --node[left]{$m_{1,0}$} (-2, 2)node[right]{$$};

    \draw (0, -4) node {$(a)$ Constant-Rate};

    \draw[lines] (4.2,0)--(10.2,0);
    \draw[lines] (7, 3.2)--(7,-3);
    
    \draw[vectors] (7, 0) --node[right]{$$} (8.5, 1.5) node[left]{$$};
    \draw[vectors] (7, 0) --node[right]{$$} (8.5, 2) node[left]{$$};
    \draw[vectors] (7, 0) --node[right]{$$} (9, 2) node[left]{$$};
    \draw[vectors] (7, 0) --node[right]{$m_{0,0}$} (9, 1.5) node[left]{$$};
    \draw[fill=black!50,opacity=0.2] (8.5, 1.5) -- (8.5, 2) -- (9, 2) -- (9,
    1.5) -- (8.5, 1.5); 
    
    \draw[vectors] (7, 0) --node[right]{$$} (8.5, -2.5)node[right]{$$};
    \draw[vectors] (7, 0) --node[right]{$$} (8.5, -2)node[right]{$$};
    \draw[vectors] (7, 0) --node[right]{$$} (9, -2)node[right]{$$};
    \draw[vectors] (7, 0) --node[right]{$m_{0,1}$} (9, -2.5)node[right]{$$};
    \draw[fill=black!50,opacity=0.2] (8.5, -2.5) -- (8.5, -2) -- (9, -2) -- (9,
    -2.5) -- (8.5, -2.5); 
    
    \draw[vectors] (7, 0) --node[left]{$$} (4.5, 1.5)node[right]{$$};
    \draw[vectors] (7, 0) --node[left]{$m_{1,0}$} (4.5, 2)node[right]{$$};
    \draw[vectors] (7, 0) --node[left]{$$} (5, 2) node[right]{$$};
    \draw[vectors] (7, 0) --node[left]{$$} (5, 1.5) node[right]{$$};
    \draw[fill=black!50,opacity=0.2] (4.5, 1.5) -- (4.5, 2) -- (5, 2) -- (5,
    1.5) -- (4.5, 1.5); 
    
    \draw (7, -4) node {$(b)$ Bounded-Rate};

   \draw[lines] (11.2,0)--(17.2,0);
    \draw[lines] (14, 3.2)--(14,-3);

    \draw[vectors] (14, 0) --node[right]{$$} (15.5, 1.5) node[left]{$$};
    \draw[vectors] (14, 0) --node[right]{$$} (15.5, 2) node[left]{$$};
    \draw[vectors] (14, 0) --node[right]{$$} (16, 2) node[left]{$$};
    \draw[vectors] (14, 0) --node[right]{$m_{0,0}$} (16, 1.5) node[left]{$$};
    
    \draw[vectors] (14, 0) --node[right]{$$} (15.5, -2.5)node[right]{$$};
    \draw[vectors] (14, 0) --node[right]{$$} (15.5, -2)node[right]{$$};
    \draw[vectors] (14, 0) --node[right]{$$} (16, -2)node[right]{$$};
    \draw[vectors] (14, 0) --node[right]{$m_{0,1}$} (16, -2.5)node[right]{$$};
    
    \draw[vectors] (14, 0) --node[left]{$$} (11.5, 1.5)node[right]{$$};
    \draw[vectors] (14, 0) --node[left]{$m_{1,0}$} (11.5, 2)node[right]{$$};
    \draw[vectors] (14, 0) --node[left]{$$} (12, 2)node[right]{$$};
    \draw[vectors] (14, 0) --node[left]{$$} (12, 1.5)node[right]{$$};

    \draw (14, -4) node {$(c)$ Multi-Rate };

  \end{tikzpicture}
}
  \vspace{-2em}
  \caption{ \label{fig:example1} Restricted Multi-mode Systems}
  \vspace{-1em}
\end{figure}

The following theorem is the key observation of the paper. 
\begin{theorem}
  \label{thmbms}
  Given a \bmms{} $\Hh = (M, n, \Rr)$, convex safety set $S$, starting state 
  $\px_0 \in \interior(S)$ and target state $\px_t \in \interior(S)$, the target
  state $\px_t$ is robustly reachable if and only if for every $\cmms{}$ in
  $\sem{\Ext(\Hh)}$ the state $\px_t$ is reachable from $\px_0$.
\end{theorem}


Alur et al.~\cite{ATW12} presented a polynomial-time algorithm to decide if a
state $\px_t$ is reachable from a starting state $\px_0$ for \cmms{}.
In particular, for starting and target states in the interior of the safety set,
they characterized a necessary and sufficient condition. 
\begin{theorem} [\cite{ATW12}]
  \label{thmcms}
  The scheduler has a winning strategy in a \cmms{} $(M, n, R)$, with convex
  safety set $S$ and starting state $\px_0 \in \interior(S)$ and target state
  $\px_t \in \interior(S)$, if and only if there is  $\vec{t} \in \Rplus^{|M|}$
  satisfying: 
  \begin{equation}\label{eq:cmms}
    \px_0(j) + \sum_{i=1}^{|M|} R(i)(j) \cdot \vec{t}(i) = \px_t(j)  \text{ for $1\leq
      j \leq n$}. 
  \end{equation}
\end{theorem}
Notice that in such a case scheduler has a strategy to reach the target state
precisely. 
The intuition behind Theorem~\ref{thmcms} is that the scheduler has a winning
strategy if and only if it is possible to reach the target state from the
starting state in using a combination of the rate vectors. 

Using Theorems~\ref{thmbms} and~\ref{thmcms} it follows that the robust
reachability problem is in co-NP. 
By reducing the validity checking problem of propositional logic formulas in
DNF, we show that the robust reachability problem for \bmms{} is indeed complete
the class co-NP. 
On a positive side, we show that the robust reachability problem for $\bmms{}$
and $\cmms{}$ is fixed parameter tractable, i.e. it is polynomial for fixed
number of variables. 
It brings us to our next key result. 
\begin{theorem}[Complexity]
  \label{thm:key-complexity}
  The robust reachability problems for \bmms{} and \cmms{} are co-NP complete. 
  However, it is fixed parameter tractable with fixed number 
  of variables. 
\end{theorem}

\section{Decidability and Complexity}
\label{sec:key_result}
This section is dedicated to the proofs of Theorem~\ref{thmbms} and Theorem~\ref{thm:key-complexity}. 
\subsection{Proof of Theorem~\ref{thmbms}}
We prove Theorem~\ref{thmbms} by showing that the condition is necessary 
and sufficient in the following two lemmas.

\begin{lemma}
\label{lem:CMS}
Given a \bmms{} $\Hh$, safety set $S$, starting state $\px_0 \in \interior(S)$ and target state
$\px_t \in \interior(S)$, the target state is not robustly reachable if there
exists a  \cmms{} $(M, n, R)$ in  $\sem{\Ext(\Hh)}$ for which $\px_t$ is not
reachable from $\px_0$.  
\end{lemma}
\begin {proof}
  
  


  From Theorem \ref{thmcms}, we have that an interior point $\px_t$ of $S$ is
  reachable iff it is in the conical hull of rates in the $\cmms{}$. 
  Let $\mathcal K$ denote this connical hull.
  Note that $\mathcal K$ is closed and that, by our assumption, $\px_t \notin
  \mathcal K$. 
  This implies that the distance $\varepsilon = \inf_{\px \in \mathcal
    K}\norm{\px_t,\px}$ between $\px_t$ and $\mathcal K$ is positive. 
  Consequently, $B_{\epsilon}(\px_t)$ and $\mathcal K$ are disjoint.
  It follows that when the environment follows the strategy to choose the rate
  $R(m)$ when presented with a mode $m$, then the scheduler cannot reach an
  $\epsilon$ ball  around $\px_t$. 
  \qed
\end{proof}


\begin{lemma} 
\label{lem:algsAreFab}
Given a \bmms{} $\Hh$, safety set $S$, starting state $\px_0 \in \interior(S)$ and target state
$\px_t \in \interior(S)$, the target state is robustly reachable if for  all \cmms{} $(M, n, R)$ in
$\sem{\Ext(\Hh)}$ the state $\px_t$ is reachable from $\px_0$. 
\end {lemma}
We give a constructive proof of this lemma by constructing an algorithm
(Algorithm~\ref{ReachabilityAlg}) giving a strategy of the player to
reach $B_\epsilon(\px_t)$  for a given $\epsilon > 0$. 

Before we elaborate on the working of the algorithm, we need to explain the idea
of projections.   
In our algorithms we sometimes represent a point $\px \in 
\Real^n$ by explicitly defining its projection towards the direction vector 
$\vv = \px_t - \px_0$ and (small) projections towards extreme rate vectors of various modes.
\begin{definition}[Projection] 
Given a \bmms{} $(M, n, \Rr)$ we say that a tuple
$(\lambda, \pi)$ is a \emph{projection}  of a point $\px$, where 
$\lambda \in \Rplus$ is the projection towards $\vv$ and 
$\pi: \Nat \times \Nat \to \Rplus$ is the projection towards extreme
rate-vectors of various modes, such that $\pi(i, j)$ is the projection towards
$j$-th vertex of the rate polytope $\Rr(m_i)$, if: 
\[
\px = \lambda \cdot \vv + \sum_{i=1}^{|M|}
\sum_{j=1}^{|\Vtx{\Rr(m_i)}|} \pi(i, j) \cdot  \Rr(i)(j).
\]
Notice that such projections are often not unique. 
We write the $(\lambda_0, \pi_0)$ for the projection such that $\lambda_0 = 0$
and $\pi_0(i, j) = 0$ for all $i, j$.
Given a projection $P = (\lambda, \pi)$ of a state $\px$ we say that $\pi(i, j)$ is
the contribution of the $j^{th}$ vertex of the rate polytope of mode $m_i$.
We also say that a vertex $\Rr(i)(j)$ does not contribute in a projection $P$
if $\pi(i, j) = 0$, while we say that a mode does not contribute in a
projection $P$  if $\pi(i, j) = 0$ holds for all corners $j$ of $\Rr(m_i)$. 
\end{definition}

Given a tolerance level of $\epsilon$, the strategy for the player to reach $B_\epsilon(\px_t)$ 
is given by Algorithm~\ref{ReachabilityAlg}. A feature of the algorithm is that the 
player selects time $\tau$ at every step. It calls function $\nextMode$ described in 
Algorithm~\ref{nextMode} to get the mode that the player chooses. Depending on 
the choice of the environment, the current point is updated. This process goes on 
until an $\epsilon$ ball around $\px_t$ is reached.
\begin{algorithm}[h]
  \KwIn{BMMS $\Hh$, starting state $\px_0$, tolerance level $\epsilon$}
  \KwOut{Reachability Algorithm}
  $\M := \displaystyle\max_{m\in M}\displaystyle\max_{\vr\in \Rr(m)}\norm{\vr}$\;
  $\px := \px_0$, the current point\;
  $\vv := \px_t-\px_0$, the reachability direction\;
  $\gamma_1 := $ shortest distance of $\px_0$ from the boundary of $S$\;
  $\gamma_2 := $ shortest distance of $\px_t$ from the boundary of $S$\;
  $\tau := \min (\epsilon/2, \gamma_1, \gamma_2)/(\M|M|)$\;
  $\sigma := $ an array, one element for each \cmms{} $\Ff$ in $\sem{\Ext(\Hh)}$ s.t $\sigma(\Ff) = \vt$ s.t. $\vt$ is a solution for the \cmms{} $\Ff$ for reachability to $\vv$\;
  Projection $P$ = ($\lambda_0, \pi_0)$ \;
  \While{$\norm{\px-\px_t} > \epsilon$}{

    $P$, $m =$ \nextMode ($P$, $\vv$, $\sigma$)\;
    $\vr = $ \getRateFromEnv($\px$, $m$, $\tau$)\;
    $\px = \px + \tau\vr$\;
    $P$ = \updateProjection ($P$, $m$, $\vr$)\;
  }
  \caption{Dynamic reachability algorithm}
  \label{ReachabilityAlg}
\end{algorithm}

\begin{algorithm}[h]
  \KwIn{Projection $P$, reachability direction $\vv$, reachability solution array $\sigma$}
  \KwOut{Projection $P$, mode $m$} 
  \While{true}{
    \If {there exists mode $m_i$ having zero contribution to Projection $P$} {
      return $P$, $m_i$\;
    }
    \Else {
      $R(m_i) := \Vtx{\Rr(m_i)}(j)$ such that $\Vtx{\Rr(m_i)}(j)$ has non-zero contribution to $P$ for each $i$\;
      \cmms {} $\Ff := (M,n,R)$ is the corresponding instance of $\Hh$\;
      $P = \reduceComp(\Ff, \sigma(\Ff), P)$\;
    }
  }
  \caption{\nextMode($P$, $\vv$, $\sigma$)}
  \label{nextMode}
\end{algorithm}
\begin{algorithm}[h]
  \KwIn{\cmms {} $\Ff = (M,n,R)$, solution time vector for the \cmms {}
    $\sigma(\Ff)$, current Projection $P = (\lambda, \pi)$}
  \KwOut{Projection $P' = (\lambda', \pi')$ s.t. contribution of one of the
    rates in $\Ff$ has been nullified} 
  $(\lambda', \pi') = (\lambda, \pi) $\;
  $k := \displaystyle\argmin_{i, \sigma(\Ff)(i) > 0} (\pi(i, R(i))/\sigma(\Ff)(i))$\;
  $\lambda' = \lambda + \pi(k, R(m_k))/\sigma(\Ff)(k)$\;
  \For {$i$ in $\set{1, 2,\ldots, |M|}$} {
    $\pi'(i, R(i)) = \pi(i, R(i)) - (\pi(k, R(k))*\sigma(\Ff)(i))/\sigma(\Ff)(k)$\;
  }
  return $(\lambda', \pi')$\;
  \caption{$\reduceComp(\Ff, \sigma(\Ff), P)$}
  \label{reduceComp}
\end{algorithm}

\begin{algorithm}[h]
  \KwIn{current Projection $P = (\lambda, \pi)$, mode $m_i$, rate $\vr$}
  \KwOut{Projection $P = (\lambda', \pi')$ with rate $r$ taken for time $\tau$}
  $(\lambda', \pi') = (\lambda, \pi)$\;
  $\vr = \sum_{j=1}^{|vert(\Rr(m_i))|}\theta_j vert(\Rr(m_i))(j)$, the convex combination of vertex of the rate polytope of mode $m_i$\;
  \For {$j$ in $\set{1,2,\ldots, |vert(\Rr(m_i))|}$} {
      $\pi'(i, j) = \theta_j*\tau$\;
    }
    return $(\lambda', \pi')$\;
    \caption{$\updateProjection (P, m, \vr)$}

    \label{updateProjection}
\end{algorithm}


%
The job of $\nextMode$ function is to nullify the contribution of a mode $m$ by expressing
the point in a different way. It calls upon $\reduceComp$ function described in 
Algorithm~\ref{reduceComp} to achieve this.
The correctness of the Algorithm~\ref{reduceComp} follows from the following
proposition.
\begin{proposition}
  Every non-negative linear combination of rates of a \cmms {} $\Ff = (M, n, R)
  \in \sem{\Ext(\Hh)}$  that reaches $\vv$ can be written as  the sum of
  a non-negative component along $\vv$  and a non-negative linear combination of
  the rates where contribution of one of the rates is $0$.   
\end{proposition}
\begin{proof}
  Given a \cmms {} $\Ff = (M, n, R)$ as an instance of $\sem{\Ext(\Hh)}$, we have
  \begin{equation} \label {soln}
    \vv = \sum_{i=1}^{|M|}\sigma(\Ff)(i).R(m_i)
  \end{equation}
  Given any non-negative linear combination of vectors in $R$,
  $\sum_{i=1}^{|M|}c_iR(m_i)$, $c_i\ge 0$, let 
  $k = \argmin_{i, \sigma(\Ff)(i) >
    0} (c_i/\sigma(\Ff)(i))$.
  The following calculations show that every non-negative linear combination of rates 
  of a \cmms {} $\Ff = (M, n, R)$  that reaches $\vv$ can be written as
  the sum of a non-negative component along $\vv$  and a non-negative linear
combination of the rates where contribution of one of the rates is $0$.  
For clarity, $\sigma(\Ff)(i)$ has been written as $\sigma_i$ below. 
\begin{eqnarray}
  \sum_{i=1}^{|M|}c_iR(m_i) &=& \sum_{i=1}^{|M|}(c_i-\frac{c_k\sigma_i}{\sigma_k} + \frac{c_k\sigma_i}{\sigma_k})R(m_i)\nonumber\\
  &=& \sum_{i=1}^{|M|}(c_i-\frac{c_k\sigma_i}{\sigma_k})R(m_i) + \sum_{i=1}^{|M|}\frac{c_k\sigma_i}{\sigma_k}R(m_i)\nonumber\\
  &=& \sum_{i=1}^{|M|}(c_i-\frac{c_k\sigma_i}{\sigma_k})R(m_i) + \frac{c_k}{\sigma_k}\sum_{i=1}^{|M|}\sigma_iR(m_i)\nonumber\\
  &=& \sum_{i=1}^{|M|}(c_i-\frac{c_k\sigma_i}{\sigma_k})R(m_i) + \frac{c_k}{\sigma_k}\vv\nonumber
\end{eqnarray}

The last step follows from Equation \ref{soln}. Note that, since 
$k = \argmin_{i, \sigma_i > 0} (c_i/\sigma_i)$, we have that 
$c_i-\frac{c_k\sigma_i}{\sigma_k} \ge 0$ holds for each $i$ and $=0$ holds for
$i=k$. 
\end{proof}
This explains the working of Algorithm~\ref{reduceComp}. We say that the
contribution of $R(m_k)$ is consumed in this process. 
Every invocation of Algorithm~\ref{reduceComp} consumes at least one corner of
one of the modes in $M$. 
Hence, it guarantees that after some finite iterations, some mode will be
consumed in the process. 
This proves the termination of Algorithm~\ref{nextMode}.

\begin{proposition}[Safety]
  All the states visited during an execution of Algorithm~\ref{ReachabilityAlg}
  are strictly inside the safety set. 
\end{proposition}
\begin{proof}
  We will demonstrate that all the point visits during a run belong to the
  safety set.
  We first claim that the point reached at any step in the algorithm, $x$, can be 
  written as the sum of a non-negative component along $\vv$ and small components
  along some rate in each of the modes. 
  Formally, 
  \begin{eqnarray}
    \px & =& \lambda \vv+ \sum_{i=1}^{|M|}t_i\vr_i, \label{form}
  \end{eqnarray}
  where  $\lambda \ge 0$, $\vr_i \in \Rr(m_i)$, $0 \le t_i \le \tau$ for all modes $m_i \in M$.   
  We prove this by induction on the number of steps. The initial point $x_0 = 0$ is trivially written in the above 
  form with $\lambda=0$, $t_i=0$ and $r_i$, any rate vector in mode $m_i$ for all $i$. If, after $j$ steps, 
  $\px  = \lambda \vv+ \sum_{i=1}^{|M|}t_i\vr_i$ with $\lambda\ge 0,\ t_i\ge 0$ for all $i$, then
  Algorithm~\ref{nextMode} ensures that $\px$ can be written in an alternative way such that the contribution of some
  mode $m_k$ is $0$ in $\px$. In this process, $\lambda$ is non-decreasing and the contribution of other modes 
  is non-increasing but always $\ge 0$. This provides $\px = \lambda_1 \vv + \sum_{i\in [|M|]\setminus\set{k}}t_i'\vr_i'$ 
  with $\lambda_1\ge \lambda, 0\le t_i'\le t_i$. The mode chosen by the player in this step is $m_k$ and the 
  time chosen is $\tau$, the new point reached is $\px' = \px + \sum_{i=1}^{|M|}t_i'\vr_i'$, where $t_k'=\tau$ and 
  $\vr_k'$ is the rate chosen by the environment in mode $m_k$. So, we again
  have $\px'$ in the form specified by Equation~\ref{form}. 

  \begin{eqnarray*}
  |\px-\lambda \vv| &=& |\sum_{i={1}}^{|M|}t_ir_i| 
 \le \sum_{i={1}}^{|M|}|t_ir_i|\\
  &\le& \sum_{i=1}^{|M|}\tau \M 
   \le \sum_{i=1}^{|M|}\frac{\min (\epsilon/2, \gamma_1, \gamma_2)}{|M|} = \min (\epsilon/2, \gamma_1, \gamma_2).
  \end{eqnarray*}
  We have used the value of $\tau$ defined in Algorithm~\ref{ReachabilityAlg}. The last equation also shows that
  the current point is inside a ball of radius $\epsilon/2$ from the point $\lambda \vv$.
  
  We now prove that $\lambda \le 1+\epsilon/(2\norm{\vv})$. We prove this by
  induction on the number of steps. Initially,
  $\lambda=0\le1+\epsilon/(2\norm{\vv})$. 
  Let $\px_j$ and $\px_{j+1}$ be the point reached after $j$ and $j+1$ steps, respectively.
  After $j$ steps, if $\lambda_j\le 1+\epsilon/(2\norm{\vv})$ and we are not inside an $\epsilon$ ball around $\px_t$, 
  then by geometry, $\lambda_j \le 1-\epsilon/(2\norm{\vv})$.
  Let $\lambda_j$ and $\lambda_{j+1}$ be the respective projection along direction $\vv$.
  Suppose, some rate $\vr$ is taken for time $\tau$. Then,
  \[
  \norm{\px_{j+1}-\px_j} = \norm{\tau r} \le \tau \norm{r}  \le \tau \M  \le
  \frac{\epsilon}{2|M|}
  \]
  So, successive points differ by a distance of at most $\epsilon/2|M| \le
  \epsilon/2$ and they are centered around $\lambda_j\vv$ and $\lambda_{j+1}\vv$
  in a ball of radius $\epsilon/2$.  
  Since $\norm{\lambda_{j+1}\vv - \lambda_j\vv} \le \epsilon$ we have that 
  \[ 
  \lambda_{j+1} \le \lambda_j + \frac{\epsilon}{\norm{\vv}} \le
  1+\frac{\epsilon}{(2\norm{\vv})}
  \]
  The last step follows from $\lambda_j \le 1-\epsilon/(2\norm{\vv})$ as argued already.
  So, we have proved by induction that $\lambda \le 1+\epsilon/(2\norm{\vv})$. Therefore, at any step in the algorithm, $\lambda \vv$
  is a convex combination of $\px_0$ and $\px_t+\epsilon\vv/(2\norm{\vv})$. Therefore, a ball of radius $\min (\epsilon/2, \gamma_1, \gamma_2)$
  lies completely inside the safety set $S$. This follows from the definition of $\gamma_1$ and $\gamma_2$. So, the algorithm is safe.
  \end{proof}
\begin{proposition}[Termination]
   The Algorithm~\ref{ReachabilityAlg} always terminates.
\end{proposition}
\begin{proof}
  We show that the algorithm terminates in finitely many steps by
  demonstrating the progress towards the target state. 
  We say that $\tau$ of mode $m$ is pumped in
  Projection $P$ at step $j$ of the algorithm, if the player chooses mode $m$ at step $j$.
  We show that there exists $\delta > 0$ such that, in every $|M|+1$ steps of
  the algorithm, the variable $\lambda$ of the current point $P$ increases by at least $\delta$. 
  Consider $|M|+1$ consecutive runs of the algorithm. 
  Since we pump a mode at every step of the
  algorithm, there is a mode $m_i$, which was chosen twice for pumping. This
  implies that at least $\tau$ of $m_i$ was consumed in $n$ steps.
  So, at least $\tau/|\Vtx{\Rr(m_i)}|$ was consumed by some vertex $\Rr(m_i)(j)$ of $\Rr(m_i)$.
  Every \cmms {} with $R(m_i)=\Rr(m_i)(j)$ guarantees a fixed increase $\delta_1$ in $\lambda$ for 
  $\tau/|vert(\Rr(m_i))|$ consumed of corner $\Rr(m_i)(j)$,
  since the solution vector $\sigma(\Ff)$ for every \cmms {} $\Ff$ is fixed. 
  Hence, $\delta$ equals the mininum of $\delta_1$ over all vertices of all
  modes is the least increase in $\lambda$. 
  This is the minimum increase in $n+1$ steps of the algorithm. 
  Hence, progress is proved.  

  Now, we show termination. Progress of $\delta$ in every $|M|+1$ iterations along with the condition 
  $\lambda \le 1+\epsilon/(2\norm{\vv})$ and increase in $\lambda$ bounded by $\epsilon$ in every step guarantee that,
  after some finite iterations, $1-\epsilon/(2\norm{\vv}) \le \lambda \le 1+\epsilon/(2\norm{\vv})$, in which case
  $\norm{\px-\lambda\vv} \le \epsilon/2$ implies $\px\in B_{\epsilon}(\px_t)$.
\end{proof}
The proof of Lemma~\ref{lem:algsAreFab} is now complete.

\subsection{Proof of Theorem~\ref{thm:key-complexity}}
With the two lemmas in place, we can proceed with proving the main complexity
results. 
We start with the positive result that states that the problem is
tractable in practice. 
\begin{theorem}
The reachability problem for \bmms{} and \mmms\ is fixed parameter tractable,
where the parameter is the number of variables. In particular, it is polynomial
for \bmms{} and \mmms\ with fixed dimension $d$.  
\end{theorem}
\begin{proof}
We first observe that, for fixed dimensions, the number of extreme points per
mode is polynomial in the size of the defining matrix. Thus, $\Ext(\Hh)$ is
polynomial in $\Hh$ for \bmms{} $\Hh$. 

For the sake of simplicity assume that the starting point is the origin and we
wish to reach position $\overline{p}$. 

From Lemmas \ref{lem:CMS} and \ref{lem:algsAreFab}, we can infer that the
existence of a $\cmms \in \sem{\Ext(\Hh)}$ for which $\overline{p}$ is not
reachable from the origin is a necessary and sufficient criterion for refuting
reachability.  
By Theorem \ref{thmcms}, for a \cmms\ $\mathcal C$ the state $\overline{p}$ is
not reachable from the origin iff it is not in the conical hull of its
vertices. 
This is the case, iff there is a hyperplane through the origin that does not
contain $\overline{p}$, such that the half-space without $\overline{p}$ it
defines contains all vertices of $\mathcal C$. 

Let $k \leq d$ be the dimension of the hull of $\Vtx{\mathcal C}$.

We now distinguish two cases.
First, assume that $\overline{p}$ is not in the hull of $\Vtx{\mathcal C}$.
In order to validate this, we can simply take $k < d$ vectors of $\mathcal C$
and validate that they are a basis of $\Vtx{\mathcal C}$.  

Now we assume that $\overline{p}$ is in the hull of
$\Vtx{\mathcal C}$. We now work in $\Vtx{\mathcal C}$. 
There are $k-1$ vectors in $\Vtx{\mathcal C}$ that define a hyperplane in this
hull, s.t. the half-space without $\overline{p}$ it defines contains all
vertices of $\mathcal C$.
\footnote{We can start with projecting the the hyperplane we started with into
  the hull of $\Vtx{\mathcal C}$, and then stepwise move the hyperplane and
  reduce its dimension. 
If the space we are left with has more than one dimension, it is clear that we
can at least change it to include one vertex. 
We can change the hyperplane to include one vector, project everything to the
subspace orthogonal to this vector, and continue. 
The modes selected can be use to define a suitable hyperplane.}

The next observation we make is that
\begin{itemize}
\item the $k<d$ spanning vectors used from extreme points of different modes are
  sufficient to establish the first case in polynomial time for the \mmms,
  because one can cheaply check that all other modes contain a vector in the
  space they span, while  $\overline{p}$ is not a linear combination of them,
  and 

\item the $k-1$ spanning vectors used from extreme points of different modes are
  sufficient to establish the second case in polynomial time for the \mmms,
  because one can cheaply check that $\overline{p}$ is not a linear combination
  of them, and all other modes contain a vector in the $k$ dimensional space
  spanned by them and $\overline{p}$, and $\overline{p}$ does not enter
  positively in the linear combination. 
\end{itemize}

Thus, it suffices to perform cheap (polynomial) tests for sets of less than $d$
vectors. 
The number of these sets is polynomial for fixed $d$. 
\qed
\end{proof}

\begin{theorem}
The reachability problem for \bmms\ and \mmms\ are co-NP complete. 
\end{theorem}
\begin{proof}
The inclusion in co-NP is implied by Lemma \ref{lem:CMS}: to refute
reachability, it is enough to guess a \cmms\ $\mathcal C$ from $\sem{\Hh}$ of an
\mmms\ or $\sem{\Ext(\Hh)}$ from a \bmms\ $\Hh$ and to verify that the target is
not reachable from $\mathcal C$. 
The verification can be performed in polynomial time from Theorem \ref{thmcms}. 

We show co-NP hardness by reducing the validity checking of propositional logic  
formulas in DNF, where each clause is a conjunction of three literals, which
refer to different propositions. 
We give a full proof for \bmms.

Given such a formula $\varphi$ with $m$ clauses 
$D_1, \dots, D_m$ and $n \geq 3$ variables $x_1, \dots, x_n$, we construct a 
\bmms\ with less than $7m+2n+3$ modes and $n+3$ variables.
We name $n$ of these variables the propositions, $x_1, \dots, x_n$, and
there are three further variables, $y_1,y_2,y_3$, which are intuitively
manipulated in three different stages of a game. 
Initially, all variables are $0$, and the goal is to reach a state, where
$y_1=y_2=1$, $y_3=n-3$, and $x_1 = x_2 = \dots = x_n = 0$. The safety set for
all variables is the interval [-1,1]. 

Given $\varphi=D_1 \vee D_2 \vee \dots \vee D_m$,
where each $D_i$ has 3 literals, we consider subclauses of $D_1, \dots, D_m$. 
Each $D_i$ has 6 non-empty subclauses. Considering the empty clause as well, we
obtain $l \leq 7m+1$ clauses $D_1, \dots, D_m, D_{m+1}, \dots, D_l$. 
Note that we do not change $\varphi$, we only need the new clauses for technical reasons. 
Let $N(D_i)=\{j \mid x_j$ or $\neg x_j$ occurs in $D_i\}$ for all $i = 1, \ldots, l$. 

The \bmms\ has only one nondeterministic mode, $m_e$, which is also the initial mode. 
Intuitively, the environment chooses the valuation of the variables in this mode.
Our \bmms\ allows all rate vectors with $\mathcal{R}(m_e)(x_i)\in [-1,1]$ for $1\leq i \leq n$,
$\mathcal{R}(m_e)(y_1) = 1$, and $\mathcal{R}(m_e)(y_2) = \mathcal{R}(m_e)(y_3) =0$.
Intuitively, the environment tries to select a valuation of the variables $x_1,
\dots, x_n$ that does not satisfy $\varphi$ in this mode, where the value $1$
refers to `true' and $-1$ refers to `false'. 
$m_e$ is the only mode with $y_1 \neq 0$. 
Given the goal, the scheduler must be in the mode $m_e$ for exactly one time unit.

For each clause in the extended set of clauses (i.e., for $i = 1,\ldots,l$), our
\bmms\ has a clause mode, $m_i$. We have: 
\begin{itemize}
 \item $\mathcal{R}(m_i)(x_j) = 1$ if $\neg x_j$ occurs in $D_i$,
 \item $\mathcal{R}(m_i)(x_j) = - 1$ if $x_j$ occurs in $D_i$,
 \item $\mathcal{R}(m_i)(x_j) = 0$ for all $j\notin N(D_i)$, and
\item  $\mathcal{R}(m_i)(y_1) = \mathcal{R}(m_i)(y_3) = 0$, and
       $\mathcal{R}(m_i)(y_2) = 1$.
\end{itemize}
Intuitively, the scheduler selects a clause from $D_1,\ldots,D_m$, and resets
the values of the three variable occurring in the clause to $0$. The role of the
additional $l-m$ clauses is to account for the capability of the environment to
select values different from $-1$ and $1$. The clause modes are the only modes
with $y_2 \neq 0$. 
Given the goal, the scheduler must be in clause modes for exactly $1$ time unit.

For each of variable $x_i$, our \bmms\ has two correction modes, $m_i^+$ and
$m_i^-$, and one empty correction node $m_0$. 
We have:
\begin{itemize}
 \item $\mathcal{R}(m_i^+)(x_i) = 1$, 
       $\mathcal{R}(m_i^+)(x_j) = 0$ for all $j\neq i$,
       
       $\mathcal{R}(m_i^+)(y_1) = \mathcal{R}(m_i^+)(y_2) = 0$, and
       $\mathcal{R}(m_i^+)(y_3) = 1$, 
       
 \item $\mathcal{R}(m_0)(x_i) = 0$, for all $i=1,\ldots,n$,
       $\mathcal{R}(m_0)(y_1) = \mathcal{R}(m_0)(y_2) = 0$, and
       $\mathcal{R}(m_0)(y_3) = 1$, and

 \item $\mathcal{R}(m_i^-)(x_i) = - 1$, 
       $\mathcal{R}(m_i^-)(x_j) = 0$ for all $j\neq i$,
       
       $\mathcal{R}(m_i^-)(y_1) = \mathcal{R}(m_i^+)(y_2) = 0$, and
       $\mathcal{R}(m_i^-)(y_3) = 1$. 
\end{itemize}
Intuitively, the scheduler resets the values of the remaining $n-3$ variables,
not covered by the clause, to $0$ using these correction modes. 
The correction modes are the only modes with $y_3 \neq 0$.
Given the goal, the scheduler must be in correction modes for exactly $n-3$ time units.

We first observe that the reachability problem is polynomial in $\varphi$.
Next, we convince ourselves that the goal is reachable if $\varphi$ is valid.

In this case, the scheduler first stays in mode $m_e$ for one time unit.
It then identifies an $i \in \{1,\ldots,m\}$ such that, for all $j \in N(D_i)$,
if $x_j > 0$ then $x_j$ is a literal of $D_i$ and if $x_j < 0$ then $\neg x_j$
is a literal of $D_i$. 
The scheduler can then apply the clause modes for $D_i$ and/or its subclauses
for together one time unit such that, after this time unit, $x_j = 0$ holds for
all $j \in N(D_i)$. 

Next, the scheduler can apply, for all $j \notin N(D_i)$ the correction mode
$m_j^+$ for $x_j$ time units if $x_j > 0$ or $m_j^-$ for $-x_j$ time units if
$x_j <0$. Given a clause $D_i$,  
$\big|\big\{j\in \{1,\ldots,n\} \mid j \notin N(D_i)\big\}\big|=n-3$, this
brings us to a point with $x_1 = \ldots = x_n = 0$, $y_1=y_2=1$, and $y_3 \in
[0,n-3]$. From there, we can apply $m_0$ for $y_3+3-n$ time units to reach the
goal. 

Finally, we have to check that, if $\varphi$ is not valid, then the goal is not reachable.
To see this, note that the $m_e$ must be scheduled for exactly one time unit.
The environment can therefore select a configuration that does not satisfy
$\varphi$ and choose rates $-1$ for `false' and $1$ for `true' for this
configuration each time $m_0$ is scheduled. 

Now let us assume that the environment follows this policy, but the goal is reached.
First we observe that the system must be for $1$ time unit in $m_e$, for $1$
time unit in clause modes, and for $n-3$ time units in correction modes. 
Clearly, some clause mode $m_i$ is used for $t$ time units, with $t \in [0,1]$.
Note that, if $D_i$ refers to a clause that is satisfied by the configuration,
then $D_i$ has at most two literals. Now we observe that 
\begin{itemize}
\item when considering the effect of the $1$ time unit in $m_e$, we have
  $\sum_{j=1}^n |x_j| = n$, 

\item when considering the $1+t$ time units the system is in $m_e$ or a clause
  mode $m_i$, we have  $\sum_{j=1}^n |x_j| \geq n - 2t$,

\item when considering the $2$ time units the system is in $m_e$ or a clause
  mode $m_i$, we have  $\sum_{j=1}^n |x_j| \geq n + t - 3$ (no $D_i$ exists
  satisfying the chosen assignment, thus $t>0$)   

\item after the complete $n-1$ time units of the run, we have  $\sum_{j=1}^n
  |x_j| \geq t > 0$. 
\end{itemize}
This provides a contradiction to having reached the goal.

The proof can easily be extended to \mmms, however we have to overcome the
exponential size of the extreme-rates for $m_e$. 
In order to achieve this, we split $y_1$ into $n$ variables $y_1^1,\ldots,
y_1^n$ and replace $m_e$ by $n$ modes $m_e^1, \ldots, m_e^n$. $\mathcal
R(m_e^i)$ has two points, where $y_1^i = 1$, $x_i \in \{-1,1\}$ and all other
$y_2=y_3=x_j = 0$ for all $j\neq i$. For the goal, we require $y_1^1=\ldots =
y_1^n = 1$ instead of $y_1 = 1$. 
The only change is that the environment now selects the values for the atomic
propositions successively instead of concurrently. 
\end{proof}

\section{Applications}
\label{apps}
In this section, we show how to apply our results for
\begin{itemize}
 \item robust schedulability---to decide if, for all $\varepsilon >0$, there is a non-Zeno control strategy, which guarantees that the system stays in an $\varepsilon$ ball around the starting point;
 \item robust stability---to decide if, for all $\varepsilon > 0$, there is a
   non-Zeno control strategy, which guarantees that the system reaches an
   $\varepsilon$ ball around the target point and then never leaves it again (possibly while staying in a convex safety set where the starting vertex and the target $\px_t$ are inner points); and
 \item robust path following---to decide if, for all $\varepsilon >0$, a given path can be followed with $\varepsilon$ precision. 
\end{itemize}

\subsection{Robust Schedulability}
For ease of notation, we assume w.l.o.g.\ that this point is the origin $\zzero$, and we assume w.l.o.g.\ that $\varepsilon<1$.
The problem has been studied before in \cite{AFMT13}, but the proof we provide here is much simpler.

Robust schedulability can be derived from robust reachability by first tweaking the reachability problem slightly, such that one execution guarantees to stay within a $\varepsilon$-ball while consuming at least one time unit.

The central idea for adjusting a system with $d$ variables $x_1,\ldots,x_d$ is to add one variable, $c$, that serves as a clock. In all rates of all modes, the rate in which this new variable progresses is $1$.
Next, we define the safety set as $S = \{(x_1,\ldots,x_d,c) \mid \forall i \leq d.\ 2d|x_i|\leq \varepsilon\}$, or any other convex set that does not constrain the values of $c$ and that constraints the values of the remaining variables to be in the $\varepsilon$ ball around $\zzero$.
We now consider the problem of reaching the point $\px_t$ with $x_1 = \ldots = x_d = 0$ and $c = 1$ with $\varepsilon$ precision.
First, when projecting away the clock $c$, the safety set alone guarantees to be in an $\varepsilon$ ball around $\zzero$, and
second, the value of $c$ must be greater than $1-\varepsilon$, which implies with the constant rate $1$ that at least $1-\varepsilon$ time units have past.

If $\px_t$ is not robust reachable from $\zzero$, then there is an $\varepsilon$, for which $\ball{\varepsilon}{\px_t}$ cannot be reached.
Thus, no strategy exists to keep the system in an $\varepsilon/d$ ball around $\zzero$ for one time unit, as this control strategy could be applied to reach the $\varepsilon$ ball around $\px_t$.
If, however, $\px_t$ is robust reachable from $\zzero$, then we can repeatedly apply such a strategy, first for $\varepsilon_1$, then for $\varepsilon_2$, and so forth, where $\varepsilon_i = 2^{-i}\varepsilon$. It is easy to see that the resulting composed strategy is non-Zeno, as all components are finite and at least one time unit passes in each component.
It is also easy to see that the error can at most add up, such that one always stays in an $\varepsilon$ ball around the starting point.

\subsection{Robust Stability} 
Obviously, reachability to $\px_t$ and robust schedulability are prerequisites for robust stability.
To see that they are also sufficient, we assuming w.l.o.g.\ that the ball $\ball{\varepsilon}{\px_t}$ is contained in the safety set $S$.
It then suffices to reach an $\px_t$ with precision $\varepsilon/2$, and then to follow a robust reachability strategy to stay in an $\varepsilon/2$ ball around the point reached.

\subsection{Robust Path Following}
To robustly follow a piecewise linear path with precision $\varepsilon$, we can simply follow the first piece with precision $\varepsilon_1$, the second with $\varepsilon_2$, and so forth, where $\varepsilon_i = 2^{-i}\varepsilon$.
Following a piecewise linear path is therefore possible with arbitrary precision if each segment can be followed individually with arbitrary precision. Conversely, if one of these segments cannot be followed with arbitrary precision, then, obviously, the complete path cannot be followed with arbitrary precision. Note that the necessary and sufficient criterion extend to infinite paths composed of an infinite sequence of segments.

Following a segment with arbitrary precision is essentially a robust reachability problem.
If the endpoint of the segment is robustly reachable from its starting point, then we can, for a given $\varepsilon$, define an convex set, where each point has distance at most $\varepsilon$ to the segment, and that contains the $\varepsilon/2$~ball around the goal. We then run Algorithm~\ref{ReachabilityAlg}.

This can be extended to piecewise smooth (continuously differentiable) paths that can be approximated arbitrarily closely by a (possibly infinite) sequence of segments, where the endpoint of each segment is reachable from its starting point. This is the case iff the derivation satisfies everywhere (where defined) the condition for robust reachability.

\section{Minimum Dwell-Time Condition}
\label{sec:dwell}
In this section, we consider an extension of robust reachability to robust
reachability with or without dwell-time or discrete sampling. We assume
w.l.o.g.\ that the minimal dwell-time or the sampling rate, respectively, is
$1$. 

\begin{theorem}
  The robust reachability problem with dwell-time requirement is decidable for
  $\bmms{}$ where all rate vectors are positive. 
\end{theorem}
\begin{proof}
W.l.o.g assume that the starting state in $\zzero$ and the target state is
$\px_t$. 
Notice that since all the rate vectors are positive, and every mode should be
taken for at least $1$ time-unit, there is a bound $K$ such that the
target state is not reachable if it is not reachable in $K$ steps.
($K$ is easy to compute.)

For robust reachability under bounded steps one can write a formula in
first-order theory of reals. 
Now the decidability of the robust reachability with dwell-time requirement for
\bmms{} with positive rate vectors follows from the decidability of the
first-order theory of reals.  
\end{proof}

\begin{theorem}
The reachability problem is $\mathsf{EXPTIME}$-hard for \mmms\ with dwell time
requirements or discrete sampling.   
\label{thm:exphard}
\end{theorem}
\begin{proof}

We prove the result by a reduction from countdown games \cite{JLS08}. 
A countdown game is a tuple $\mathcal{G}=(N, T,$ $(n_0,B_0))$, where $N$ is a
finite set  of nodes, $T \subseteq N \times \mathbb{N}_{>0} \times N$ a set of
transitions, and $(n_0,B_0) \in N \times \mathbb{N}_{>0}$ is the initial
configuration. 
The states of a countdown game, also called its configurations, are
$N \times \{0,1, \ldots, B_0\}$. 

From any configuration $(n,B)$, Player 1 chooses a number
$l \in \mathbb{N}_{>0}$ such that there exists a transition $(n, l, n') \in T$  
with $l \leq C$.
Among all the available transitions of the form $(n,l,n')$, Player 2 
selects an appropriate transition $(n,l,n'') \in T$. 
The new configuration is then $(n'', C-l)$.

Player 1 wins when a configuration $(n,0)$ is reached, and otherwise loses when a configuration $(n,C)$ is reached where Player 1 cannot move.
This is the case when, for all outgoing transitions $(n,l,n')\in T$, we have $l > C$.
W.l.o.g., we assume that there are no transitions $(n,l,n) \in T$ for any $l \in \mathbb{N}_{>0}$.

We now translate this game into a sampled robust reachability problem, where the scheduler takes the role of Player~1, while the environment takes the role of Player 2.

The translation uses $|N|+1$ variables, a variable $B$ reflecting the remaining time budget and a variable $n$ for each element $n\in N$.
Being in state $(n,C)$ in the countdown game is intuitively represented by $B=C$, $n=1$, and $n' = 0$ for all states $n' \neq n$.
The initial state is given by $B=B_0$, $n_0=1$, and  $n = 0$ for all states $n \neq n_0$, i.e., by the state representing the initial configuration $(n_0,B_0)$.
The target is~$\zzero$.
The safety set is described by $n \in [-0.5,1.5]$ for all $n \in N$ and $B \in [-0.5, B_0+1]$.

The rates Player 1 selects become the modes of our $\mmms$.
Thus, we have a mode $l$ for each $l \in \mathbb N_{>0}$, for which a transition $(n,l,n') \in T$ exists.
The selection of the concrete transition by Player 2 becomes the choice of the mode by the environment.
We therefore have, for a given mode $l$, one rate vector for each transition $(n,l,n') \in T$, where the rates are $n=-1$, $n' = 1$,  $B = -l$, and $n'' = 0$ for all $n'' \in N \smallsetminus \{n,n'\}$.

Before we describe how to translate (winning) strategies, we first note that, from each translation of a configuration, the scheduler cannot make a move of length $\geq 2$. 
We first replace the target vertex by a the target region $B=0$.
For this target region, there is a simple 1:1 translation between the moves and states for the countdown game and the reachability game, where each move $l$ of Player 1 in the countdown game corresponds to the move $(l,1)$ of the scheduler, while every move $(n,l,n')$ of Player 2 corresponds to the environment selecting the corresponding rate.

To return to normal reachability, we add, for each node $n\in N$, a mode $n$.
This mode has only one rate, with  $B=0$, $n = - 1$, and $n' =0$ for all $n' \neq n$. 
Note that such a mode $n$ can only be applied from states that encode $(n,C)$, and it can only be applied with duration $1$.
Once such a mode is applied, no further mode (of either type) can be applied in the future, as one variable $n'' \in N$ would afterwards have the value $-1$.

Now, a winning for Player 1 corresponds to winning strategy of the scheduler that ends by applying such a mode.
This closes the proof for discrete sampling.

To expand this to dwell time, we sharpen the bounds for the safety set to $n \in [-\varepsilon,1+\varepsilon]$ for all $n \in N$ and $B \in [-\varepsilon, B_0+1]$ for some $\varepsilon < \frac{1}{8B_0}$.
Now, if Player 1 wins, then the scheduler wins with the same strategy as above.
If Player~2 wins, Player $1$ is stuck in $<B_0$ move pairs.
When the environment mimics such a strategy (until Player $1$ is stuck) then the game reaches a position, where each variable value is in a $2(B_0-1) \varepsilon < \frac{1}{4}$ range around the value it would have, had the scheduler played a duration of $1$ for each move.
Thus, the scheduler can, at most, play a ``node mode'' $n \in N$ once, but it cannot reduce the value of $B$ without leaving the safety region.
\end{proof}

For discrete sampled schedulers, we can easily show inclusion in $\mathsf{EXPTIME}$ by exploring the complete state-space.
To do this, we can proceed in two steps. In a first step, we expand all values in the problem setting to integers by multiplying every value with the least common multiple of all denominators. (Note that this is a polynomial time reduction.)
Then we can be sure that all values are at integer points, and we can simply explore the complete state-space, which is exponential in the setting.
As the lower bound is inherited from the previous proof, we get:

\begin{corollary}
The robust reachability problem with discrete sampling is $\mathsf{EXPTIME}$-complete.
\end{corollary}

\section{Generalized Models}
\label{sec:generalized}
In this section we consider generalization of the \bmms{} by adding structure to
the model using Alur-Dill style~\cite{alurDill94}  clock variables,
i.e. variables with rate $1$ in every mode.
In the resulting model only clock variables can occur on the transitions where
they can be compared against natural numbers or can be reset. 
All other non-clock variables will behave like \bmms{}. 
We show that for \bmms{} with clock the robust reachability problem is
undecidable for \bmms{} with 2 variables and 1 clock, and \bmms{} with 1
variable and 2 clocks.

We prove the undecidability of this problem by giving a reduction from the
halting problem for two-counter machines. 
Formally, a two-counter machine (Minsky machine) $\Aa$ is a tuple $(L, C)$ where:
${L = \set{\ell_0, \ell_1, \ldots, \ell_n}}$ is the set of
instructions. There is a distinguished terminal instruction  $\ell_n$ called
HALT. 
${C = \set{c_1, c_2}}$ is the set of two \emph{counters};
the instructions $L$ are one of the following types:
\begin{enumerate}
\item 
  (increment $c$) $\ell_i : c := c+1$;  goto  $\ell_k$,
\item (decrement $c$) $\ell_i : c := c-1$;  goto  $\ell_k$,
\item (zero-check $c$) $\ell_i$ : if $(c >0)$ then goto $\ell_k$
  else goto $\ell_m$,
\item (Halt) $\ell_n:$ HALT.
\end{enumerate}
where $c \in C$, $\ell_i, \ell_k, \ell_m \in L$.

A configuration of a two-counter machine is a tuple $(\ell, c, d)$ where
$\ell \in L$ is an instruction, and $c, d$ are natural numbers that specify the value
of counters $c_1$ and $c_2$, respectively.
The initial configuration is $(\ell_0, 0, 0)$.
A run of a two-counter machine is a (finite or infinite) sequence of
configurations $\seq{k_0, k_1, \ldots}$ where $k_0$ is the initial
configuration, and the relation between subsequent configurations is
governed by transitions between respective instructions.
The run is a finite sequence if and only if the last configuration is
the terminal instruction $\ell_n$.
Note that a two-counter  machine has exactly one run starting from the initial
configuration. 
The \emph{halting problem} for a two-counter machine asks whether 
its unique run ends at the terminal instruction $\ell_n$.
It is well known that the halting problem for two-counter machines is
undecidable. 

\begin{theorem}
  The robust reachability problem is undecidable for \bmms{} with $2$ variables
  and $1$ clock.
\end{theorem}
\begin{proof}
  For the sake of simplicity of presentation we prove the undecidability of the
  exact reachability problem. 
  The proof can be adapted to robust reachability case.
  Given a Minsky machine we construct a structured BMS ${\mathcal H}$ with 2
  variables and a single clock that is reset on every transition.  
  The clock is used in a simple way just to ensue that at each mode exactly $1$
  unit of time is spent by the scheduler. 
  We use two variables $y$ and $z$ to encode the values of the two counters
  $c_1$ and $c_2$, and one mode corresponding to each location of the Minsky
  machine. 
  For each zero check instruction we further use five extra modes and a special
  target mode $\Tt$ depicted by a double circle. 
  Our goal is to reach mode ${\mathcal T}$ with $y=z=0$.

  The simulation of the increment and decrement instruction is straightforward. 
  In an increment $c_1$ location the rate is given by $(1,0)$, while in
  decrement location the rate is give by $(-1, 0)$. 
  Clock variables are used to ensure that exactly one time unit is spent in each
  such mode. 


 
\begin{figure}[t]
\begin{center}
\scalebox{0.7}{
\begin{tikzpicture}[->,>=stealth',shorten >=1pt,auto,node distance=1.8cm,
  semithick]
  \tikzstyle{every state}=[minimum size=3em,rounded rectangle]
  \node[initial,initial text={},state] at (-3, 0) (A) {$\begin{array}{c}\ell_i \\ (0,0) \end{array}$} ;
  
  \node[state] at (0,1.5) (B) {$\begin{array}{c}NZ \\ \{(0,0), (0,-100)\} \end{array}$} ;
  \node[state] at (0,-1.5) (B1) {$\begin{array}{c}Z \\ \{(0,0), (0,-100)\} \end{array}$} ;
  \node[state, player2] at (6,1.5) (C1) {$\ell_2$};
  \node[state, player2] at (6,-1.5) (D1) {$\ell_1$};
  
  \node[state]at (0, 3.5) (TA) {$\begin{array}{c}O \\ (0,100) \end{array}$} ;
  \node[state] at (3,3.5) (TB) {$\begin{array}{c}N \\ (0,-1) \end{array}$} ;
  \node[state,accepting] at (6,3.5) (TC) {$\begin{array}{c}T \\ (-1,0) \end{array}$} ;
  
  \draw[dashed,draw=gray,rounded corners=10pt] (-1.2,2.5) rectangle (7.5,5.3);
  \node[rotate=90,color=gray] at (-1.6, 3.5) (N) {$Chk_1$};

  \path (TA) edge node[above] {$ x = 1?$} node[below]{$x:=0$} (TB);
  \path (TB) edge node[above] {$ x = 1?$} node[below]{$x:=0$} (TC);
 \path (TB) edge[loop above] node[above] {$x = 1?, x:=0$} (TB);
 \path (TC) edge[loop above] node[above] {$x = 1?, x:=0$} (TC);

 \node[state] (PA) at (0, -3.5) {$\begin{array}{c}O \\ (0,100) \end{array}$} ;
 \node[state,accepting] at (6,-3.5) (PC) {$\begin{array}{c}T \\ (-1,0) \end{array}$} ;
 \path (PA) edge node[above] {$ x = 1?, x:=0$} (PC);
 \path (PC) edge[loop above] node[above] {$ x = 1?, x:=0$} (C);

 \node[rotate=90,color=gray] at (-1.4, -3.5) (N2) {$Chk_2$};
 \draw[dashed,draw=gray,rounded corners=10pt] (-1.2,-4.5) rectangle (7.5,-2.5);

 \draw[rounded corners] (A) -- +(0, 1.5) --  node[above] {$x=1$}
 node[below] {$x:=0$}(B);  
 \draw[rounded corners] (A) -- +(0, -1.5) --  node[above] {$x=1$}
 node[below] {$x:=0$}(B1);  
 
 \path (B) edge node[above] {$x = 1?, x:=0$} (C1);
 \path (B1) edge node[above] {$x=1?, x:=0$} (D1);    
 
 \path (B) edge node[left] {$x = 1?$} node[right] {$x:=0$} (TA);
 \path (B1) edge node[left] {$x = 1?$} node[right] {$x:=0$} (PA);

\end{tikzpicture}}
\caption{Simulation of Zero Check instruction }
\label{zero-check-c2}
\end{center}
\end{figure}
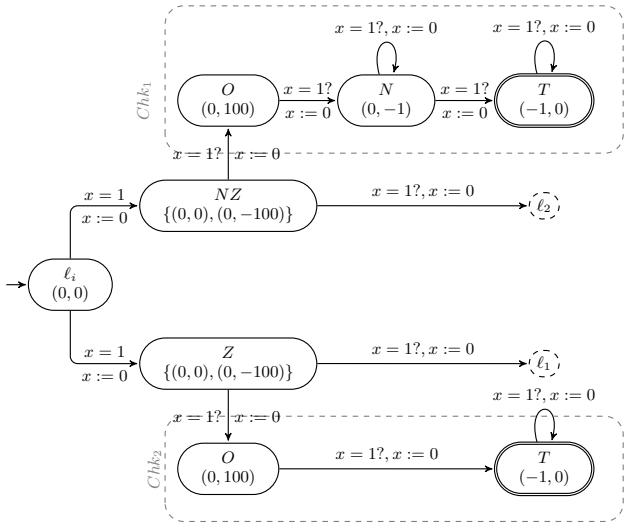  
The Zero Check Instruction is simulated using the widget shown in Figure~\ref{zero-check-c2}.
The scheduler non-deterministically guesses if $c_2$ is zero or not, by going  
to one of the locations $Z,NZ$. The values of variables $y,z$ remain unchanged.
Assume that scheduler chose $NZ$. 
The environment can now allow the scheduler to continue his simulation by either
giving  the rate $(0,0)$,  or check his guess by giving the  rate $(0, -100)$.    
If the rate $(0,0)$ is obtained, the scheduler's best strategy is to  goto
$\ell_2$, otherwise, the scheduler must go to  $Chk1$. 
The first thing that happens in the gadget $Chk1$ is the variable $z$ regaining
its previous value by adding 100. If the scheduler's choice of $c_2$ being
non-zero was incorrect, then when  the location  $T$ is reached, we have
$z=-1$. 
There is then no way to reach  the target mode $T$ with valuation $y=0,z=0$. 

In a similar way, the environment can check if the scheduler guessed that the
counter $c_2$  is zero, by giving the rate $(0, -100)$ at the location $Z$. 
In this case, the best strategy for scheduler is to goto the gadget $Chk2$. 
The first thing that happens in $Chk2$ is for variable $z$ to regain 
its previous value by adding $100$. 
If the guess of $c_2$ being $0$ was correct, then the scheduler can reach $T$
with $y=z=0$. 
However, if the guess was wrong, scheduler can obtain $z=0$, and will lose. 

If the two counter machine halts, and the scheduler simulates all the
instructions correctly, then it is possible to reach a mode $T \in {\mathcal T}$
with $y=z=0$, or the mode $Halt$ is reached. 
It is straightforward to see that the location $Halt$ is reached iff the two
counter machine halts and scheduler simulates all instructions correctly.  
From the $Halt$ mode we add an outgoing transition from where it is always
possible for the scheduler  to reach a mode $T \in {\mathcal T}$ with $y=z=0$. 
The proof is now complete.
\end{proof}
The proof of the following theorem is also via a reduction from the Minsky
machines and is slightly more involved than the previous theorem. 
However, due to space limitation, we have moved the proof to the appendix.  
\begin{theorem}
  \label{reach-dec-thm}
  The robust reachability problem is undecidable for \bmms{} with $1$ variable
  and $2$ clocks. 
\end{theorem}
\section*{Acknowledgments} 
We thank Rajeev Alur, Salar  Moarref and Vojtech Forejt for the discussions
related to some aspects of this work.

\bibliographystyle{plain}
\bibliography{papers}

\appendix

\subsection*{Proof of Theorem~\ref{reach-dec-thm}}

We prove the undecidability by constructing a structured BMS ${\mathcal H}$ with 2 clocks and one variable 
that simulates the 2 counter machine. We prove that the scheduler has a winning strategy 
to reach $x \in B_{\Delta}(p)$  iff the two counter machine halts. 
Our construction of ${\mathcal H}$ is such that we have a gadget corresponding to each instruction 
in the two counter machine.  We consider $p=7$, and 
$0 < \Delta < 1$ as given.  Modes in the target set ${\mathcal T}$ 
are denoted by a double circle. 

Let the single variable be denoted $z$, and let $x,y$ be the clocks. 
On entry into any gadget, the value of the variable $z$ is $5-\frac{1}{2^{c_1}3^{c_2}}$ where $c_1, c_2$ are the 
current values of the two counters, and the clocks $x,y$ are zero. 
\begin{enumerate}
\item Simulation of an increment instruction $\ell_i : c_1 := c_1+1$;  goto  $\ell_k$. 

The gadget simulating the increment $c_1$ instruction can be seen in Figure \ref{inc-c1}.
The gadget is entered with $z=5-\frac{1}{2^{c_1}3^{c_2}}, x=y=0$. The 
locations in the gadget contain the name of the location as well as 
the rate (possibly a set of rates, or an interval of rates)  of the variable $z$, as the case may be. 
Let us denote by $old$ the value $\frac{1}{2^{c_1}3^{c_2}}$. 
A non-deterministic amount of time is spent at location $\ell_i$. 
The ideal time to be spent here is $\frac{old}{2}$, so that 
 $z$ is updated from $5-old$ to $5-\frac{old}{2}$, reflecting 
 the correct new counter values. $y$ is reset on 
 going to location $A$. A time of one unit is spent at location $A$.
 The value of $x$ is unchanged during this process due to the self loop 
 on $A$. There are three possible rates that the environment can 
 give to the scheduler, namely 100, -100 or 0 at location $A$. 
 The scheduler can go to any of the gadgets $C_{>}, C_{<}$ or to the location  $\ell_k$.

\begin{figure}[h]
\begin{center}
\scalebox{0.7}{
\begin{tikzpicture}[->,>=stealth',shorten >=1pt,auto,node distance=1.8cm,
  semithick]
  \tikzstyle{every state}=[minimum size=3em,rounded rectangle]
  
  \node[initial,initial text={},state] (A) {$\begin{array}{c}\ell_i \\ 1 \end{array}$} ;
   \node[state] at (3.5,0) (B) {$\begin{array}{c}A \\ \{100, -100, 0\} \end{array}$} ;
\node[state, player2] at (1,3) (chk1) {$C_{>}$};
\node[state, player2] at (1,-3) (chk2) {$C_{<}$};
\node[state] at (5,2) (cont) {$\ell_k$};
      \path (A) edge node[above] {$0< x< 1?$} (B);
      \path (A) edge node[below] {$y:=0$} (B);
        \path (B) edge[loop below]  node[above] {$x=1?$}node[below]{$x:=0$} (B);
                
  \path (B) edge  node[left]{$y=1?$} node[right]{$y:=0$} (chk1);
  \path (B) edge  node[left]{$y=1?$} node[right]{$y:=0$} (chk2);
  \path (B) edge  node[left]{$y=1?$} node[right]{$x,y:=0$} (cont);
  
  \end{tikzpicture}}
\caption{Simulation of increment $c_1$ instruction }
\label{inc-c1}
\end{center}
\end{figure}
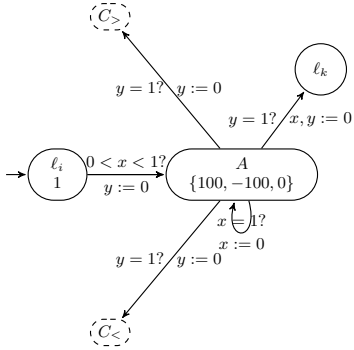

Assume that the time spent at $\ell_i$ is $\frac{old}{2}+\epsilon$ for some $\epsilon \geq 0$.
In this case, $x=\frac{old}{2}+\epsilon,y=0$ and $z=5-\frac{old}{2}+\epsilon$. The environment
can force a check of the scheduler and catch his mistake, by choosing a rate of 100 at location $A$.
 This would make $z=105-\frac{old}{2}+\epsilon$. If the scheduler wants to win, 
 he must reach a mode in $\mathcal{T}$, with the value of $z$ 
 in $B_{\Delta}(7)$. The best thing for the scheduler to do at this point
 is to choose  $C_{>}$ as his next location, since 
 it allows the value of $z$ to come back to $5-\frac{old}{2}+\epsilon$. 
 If the scheduler chooses to go to $C_{<}$, he will be worse off, making $z$ even bigger, 
 and if he chooses $\ell_k$, the environment can make sure that 
 the scheduler never wins by choosing the rate 100 in all future gadgets.  

Lets thus assume that the scheduler chooses to goto the gadget $C_{>}$ in Figure \ref{inc-c1-more}. 
On entry, we have $x=\frac{old}{2}+\epsilon$, $y=0$ and $z=105-\frac{old}{2}+\epsilon$. 
At location $O$, the value of $x$ remains unchanged, $y$ grows to 1 and is reset, 
and $z$ becomes $z=5-\frac{old}{2}+\epsilon$. At location $B$, a 
time $1-\frac{old}{2}-\epsilon$ is spent, obtaining $x=0, y=1-\frac{old}{2}-\epsilon$ 
and $z= 5-\frac{old}{2}+\epsilon-1(1-\frac{old}{2}-\epsilon)=4+2\epsilon$.
If $\epsilon >0$ and $2\epsilon > \Delta$, then the scheduler has already lost, since 
adding 3 more to $z$ at location $C$ does not help.  
Consider now the case that $\epsilon >0$ and $\Delta-2\epsilon = \kappa >0$. 
At location $D$, a time  of one unit is spent, and the environment can choose a 
rate as close to $\Delta$ as he wants : in particular, he can choose a rate 
that is larger than $\kappa$, making the value of $z=4+2\epsilon+\kappa+\zeta$, for some 
$\zeta >0$. This means the scheduler 
can never reach a point in the ball $B_{\Delta}(7)$, 
even after adding 3 to $z$ at location $C$. 

If $\epsilon=0$, then irrespective of the rate $\kappa \in (0, \Delta)$ chosen by the environment, 
the value of $z$ is $7+\kappa \in B_{\Delta}(7)$, after adding 3 to $z$ at location $C$. 
Thus, if the scheduler made no mistake, 
he reaches a point inside the chosen ball.

\begin{figure}[h]
\begin{center}
\scalebox{0.7}{
\begin{tikzpicture}[->,>=stealth',shorten >=1pt,auto,node distance=1.8cm,
  semithick]
  \tikzstyle{every state}=[minimum size=3em,rounded rectangle]
  
  \node[initial,initial text={},state] (A) {$\begin{array}{c} O\\ -100 \end{array}$} ;
        \path (A) edge[loop above]  node[above] {$x=1?$}node[below]{$x:=0$} (A);
   \node[state] at (3.5,0) (B) {$\begin{array}{c}B \\ -1 \end{array}$} ;
\path (A) edge  node[above] {$y=1?$}node[below]{$y:=0$} (B);
\node[state] at (1,-3) (C) {$\begin{array}{c}C \\ 3 \end{array}$} ;
\node[state] at (3.5,-3) (D) {$\begin{array}{c}D \\ (0, \Delta) \end{array}$} ;
\node[accepting,state] at (-2,-3) (T) {$\begin{array}{c}T \\ 0 \end{array}$} ;

\path (B) edge  node[above] {$x=1?$}node[below]{$x:=0$} (D);
\path (D) edge  node[above] {$x=1?$}node[below]{$x:=0$} (C);
\path (C) edge  node[above] {$x=1?$}node[below]{$x:=0$} (T);
  \end{tikzpicture}}
\caption{The gadget $C_{>}$}
\label{inc-c1-more}
\end{center}
\end{figure}
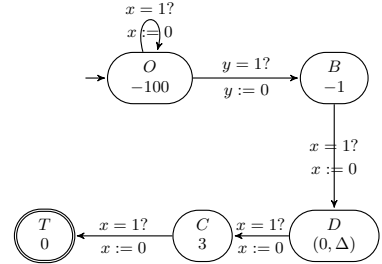

Now consider the case when the scheduler spends an amount of time $\frac{old}{2}-\epsilon$, 
for some $\epsilon \geq 0$ at location $l_i$ in Figure \ref{inc-c1}. Then we have 
$x=\frac{old}{2}-\epsilon,y=0$ and $z=5-\frac{old}{2}-\epsilon$. At location $A$ in Figure \ref{inc-c1}, as 
seen above, the environment can assign any of the rates 100, -100 or 0 to the scheduler.
If the environment wishes to catch the scheduler's mistake, a rate of -100 will be assigned. The scheduler, 
if he chooses to goto $C_{>}$ or $Go$, will surely lose, since 
the value of $z$ will decrease further, and will never reach a value in $B_{\Delta}(7)$; 
likewise, if the scheduler chooses $Go$, the environment can forever give a rate of -100. 
The best choice for scheduler is therefore, to pick $C_{<}$. The gadget $C_{<}$ is 
given in Figure \ref{inc-c1-less}.  
 
 On entry to $C_{<}$, we have $x=\frac{old}{2}-\epsilon,y=0$ and $z=-95-\frac{old}{2}-\epsilon$. 
 At location $O$, the value of $x$ remains unchanged, $y$ grows to 1 and is reset, 
and $z$ becomes $z=4-\frac{old}{2}-\epsilon$. At location $B$, a 
time $1-\frac{old}{2}+\epsilon$ is spent, obtaining $x=0, y=1-\frac{old}{2}+\epsilon$ 
and $z= 5-old$. A time $\frac{old}{2}-\epsilon$
is spent at location $C$, obtaining $z=5-old+2(\frac{old}{2}-\epsilon)=5-2\epsilon$. 

At location $E$, a time of one unit is spent, and the environment can choose a rate in $(-\Delta,0)$.
Consider the case when $\epsilon >0$ and $5-2\epsilon   < 5-\Delta$. In this case, scheduler 
has already lost the game, since spending one unit at location $D$
will only give $z=7-2\epsilon < 7-\Delta$. 
  However, if $7-2\epsilon > 7-\Delta$, 
let $\kappa=\Delta - 2\epsilon >0$. Environment can then choose a rate $-\kappa-\zeta \in (-\Delta,0)$, for $\zeta > 0$.
Then $z=7-2\epsilon-\kappa+\zeta=7-\Delta-\zeta < 7-\Delta$. This would 
result in scheduler losing. However, if $\epsilon=0$, then 
for any $-\kappa \in (-\Delta,0)$, the value of $z$ is $7-\kappa \in B_{\Delta}(7)$.

\begin{figure}[h]
\begin{center}
\scalebox{0.7}{
\begin{tikzpicture}[->,>=stealth',shorten >=1pt,auto,node distance=1.8cm,
  semithick]
  \tikzstyle{every state}=[minimum size=3em,rounded rectangle]
  
  \node[initial,initial text={},state] at (-2,0) (A) {$\begin{array}{c} O\\ 99 \end{array}$} ;
        \path (A) edge[loop above]  node[above] {$x=1?$}node[below]{$x:=0$} (A);
   \node[state] at (0.3,0) (B) {$\begin{array}{c}B \\ 1 \end{array}$} ;
\path (A) edge  node[above] {$y=1?$}node[below]{$y:=0$} (B);
\node[state] at (2.5,0) (C) {$\begin{array}{c}C \\ 2 \end{array}$} ;
\node[state] at (0,-3) (D) {$\begin{array}{c}D \\ 2 \end{array}$} ;
\node[state] at (2.5,-3) (E) {$\begin{array}{c}E \\ (-\Delta,0) \end{array}$} ;
\node[accepting,state] at (-2.5,-3) (T) {$\begin{array}{c}T \\ 0 \end{array}$} ;

\path (B) edge  node[above] {$x=1?$}node[below]{$x:=0$} (C);
\path (C) edge  node[left] {$x=1?$}node[right]{$x:=0$} (E);
\path (E) edge  node[above] {$x=1?$}node[below]{$x:=0$} (D);
\path (D) edge  node[above] {$x=1?$}node[below]{$x:=0$} (T);
  \end{tikzpicture}}
\caption{The gadget $C_{<}$}
\label{inc-c1-less}
\end{center}
\end{figure}
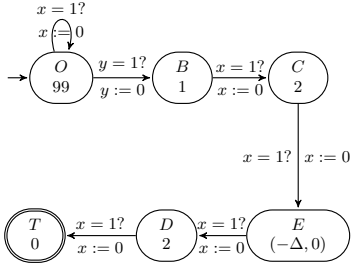

The only remaining case is when the scheduler indeed picks the correct delay of $\frac{old}{2}$ 
at location $\ell_i$ in Figure \ref{inc-c1}. In this case, as seen above, 
the rates 100, -100 chosen by the environment does not affect the scheduler. 
In both these cases, scheduler has a winning strategy of choosing to go to 
one of $C_{<}, C_{>}$ and reach a value of $z$ in the chosen ball.  
If $\epsilon=0$, and the environment picks the rate 0 at location $A$, then 
the best strategy for the scheduler is to select  $\ell_k$, which 
marks the continuation of the simulation of the two counter machine. 
As expected, we will indeed have on entry into $l_k$, $x=y=0$ and $z=5-\frac{old}{2}$,
marking the correct simulation of the increment $c_1$ instruction.
  
\item Simulation of a decrement instruction $\ell_i : c_1 := c_1-1$;  goto  $\ell_k$. 

  The construction of gadgets for the decrement instruction is similar to that of the increment 
  instruction. 
         
\begin{figure}[h]
\begin{center}
\scalebox{0.7}{
\begin{tikzpicture}[->,>=stealth',shorten >=1pt,auto,node distance=1.8cm,
  semithick]
  \tikzstyle{every state}=[minimum size=3em,rounded rectangle]
  
  \node[initial,initial text={},state] (A) {$\begin{array}{c}l_i \\ -\frac{1}{2} \end{array}$} ;
   \node[state] at (3.6,0) (B) {$\begin{array}{c}A \\ \{100, -100, 0\} \end{array}$} ;
\node[state, player2] at (1,3) (chk1) {$D_{>}$};
\node[state, player2] at (1,-3) (chk2) {$D_{<}$};
\node[state] at (5,2) (cont) {$l_k$};
      \path (A) edge node[above] {$0< x \leq 1?$} (B);
      \path (A) edge node[below] {$y:=0$} (B);
        \path (B) edge[loop below]  node[above] {$x=1?$}node[below]{$x:=0$} (B);
                
  \path (B) edge  node[left]{$y=1?$} node[right]{$y:=0$} (chk1);
  \path (B) edge  node[left]{$y=1?$} node[right]{$y:=0$} (chk2);
  \path (B) edge  node[left]{$y=1?$} node[right]{$x,y:=0$} (cont);
  
  \end{tikzpicture}}
\caption{Simulation of decrement $c_1$ instruction }
\label{dec-c1}
\end{center}
\end{figure}
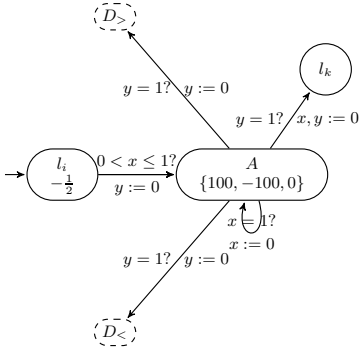

The ideal amount of time to be spent by scheduler at $\ell_i$ is $2old$.  In this case, 
$z=5-2old,x=2old,y=0$. Assume that the time spent at $\ell_i$ is 
$2old+\epsilon$, for some $\epsilon >0$. Then we have $z=5-2old-\frac{\epsilon}{2}$. 
At location $A$, the environment picks one of the 3 rates 100, -100, 0. 
If he wants to force a check on the environment, he picks the rate 100, 
making $z=105-2old-\frac{\epsilon}{2}$, $x=2old+\epsilon,y=0$. 
As seen in the case of the increment gadget, the best strategy for the scheduler is to pick the gadget $D_{>}$. 

\begin{figure}[h]
\begin{center}
\scalebox{0.7}{
\begin{tikzpicture}[->,>=stealth',shorten >=1pt,auto,node distance=1.8cm,
  semithick]
  \tikzstyle{every state}=[minimum size=3em,rounded rectangle]
  
  \node[initial,initial text={},state] (A) {$\begin{array}{c} O\\ -100 \end{array}$} ;
        \path (A) edge[loop above]  node[above] {$x=1?$}node[below]{$x:=0$} (A);
   \node[state] at (3.5,0) (B) {$\begin{array}{c}B \\ -1 \end{array}$} ;
\path (A) edge  node[above] {$y=1?$}node[below]{$y:=0$} (B);
\node[state] at (3.5,-3) (C) {$\begin{array}{c}C \\ 3 \end{array}$} ;
\node[state] at (1,-3) (D) {$\begin{array}{c}D \\ (0, \Delta) \end{array}$} ;
\node[accepting,state] at (-2,-3) (T) {$\begin{array}{c}T \\ 0 \end{array}$} ;

\path (B) edge  node[above] {$x=1?$}node[below]{$x:=0$} (C);
\path (C) edge  node[above] {$x=1?$}node[below]{$x:=0$} (D);
\path (D) edge  node[above] {$x=1?$}node[below]{$x:=0$} (T);
  \end{tikzpicture}}
\caption{The gadget $D_{>}$}
\label{dec-c1-more}
\end{center}
\end{figure}

Entry into $D_{>}$ is made with $z=105-2old-\frac{\epsilon}{2}$, $x=2old+\epsilon,y=0$. One unit 
of time is spent at $O$, obtaining $z=5-2old-\frac{\epsilon}{2}$, $x=2old+\epsilon,y=0$. 
At $B$ a time $1-2old-\epsilon$ is spent, obtaining $z=5-2old-\frac{\epsilon}{2}-1(1-2old-\epsilon)=4+\frac{\epsilon}{2}$.
A time of one unit is spent at $C$ obtaining $z=7+\frac{\epsilon}{2}$. 
Likewise, a time of one unit is spent at $D$, obtaining $z=7+\frac{\epsilon}{2}+\kappa$, where 
$\kappa \in (0, \Delta)$. If $\frac{\epsilon}{2} > \Delta$, then clearly, scheduler has already lost the game.
If $\Delta-\frac{\epsilon}{2} >0$, then $\kappa \in (0, \Delta)$ can be chosen such that $\kappa > \Delta-\frac{\epsilon}{2}$ 
such that the value of $z \notin B_{\Delta}(7)$. Note that if $\epsilon=0$, this is not possible, and 
scheduler can indeed reach  $z \in B_{\Delta}(7)$.

Now consider the case when scheduler spends a time  $2old-\epsilon$, for some $\epsilon >0$ at $\ell_i$ in Figure \ref{dec-c1}. 
 Then we have $z=5-2old+\frac{\epsilon}{2}$. Again, the environment can choose the rate -100 at location $A$, 
 and the scheduler's best strategy is to enter gadget $D_{<}$. Entry into $D_{<}$ happens with 
          $z=-95-2old+\frac{\epsilon}{2},x=2old-\epsilon,y=0$. 
     One unit of time is spent at $O$ obtaining $z=5-2old+\frac{\epsilon}{2},x=2old-\epsilon,y=0$. 
     A time $1-2old+\epsilon$  is spent at location $B$ obtaining 
     $z=5-2old+\frac{\epsilon}{2}-1(1-2old+\epsilon)=4-\frac{\epsilon}{2}$. 
     One unit of time is spent at $C$ obtaining $z=7-\frac{\epsilon}{2}$. 
    Spending one unit at location $D$  with a rate $-\kappa \in (-\Delta,0)$ 
  gives $z= 7-\frac{\epsilon}{2}-\kappa$. 
 If $\frac{\epsilon}{2} > \Delta$, then the scheduler has already lost. If $\Delta -\frac{\epsilon}{2} >0$,
 then the environment can always choose  $-\kappa \in (-\Delta,0)$ such that 
 $z=7-\frac{\epsilon}{2}-\kappa < 7-\Delta$. Clearly, if $\epsilon=0$, this 
 is not possible. and scheduler wins.

\begin{figure}[h]
\begin{center}
\scalebox{0.7}{
\begin{tikzpicture}[->,>=stealth',shorten >=1pt,auto,node distance=1.8cm,
  semithick]
  \tikzstyle{every state}=[minimum size=3em,rounded rectangle]
  
  \node[initial,initial text={},state] at (-2,0) (A) {$\begin{array}{c} O\\ 100 \end{array}$} ;
        \path (A) edge[loop above]  node[above] {$x=1?$}node[below]{$x:=0$} (A);
   \node[state] at (0.3,0) (B) {$\begin{array}{c}B \\ -1 \end{array}$} ;
\path (A) edge  node[above] {$y=1?$}node[below]{$y:=0$} (B);
\node[state] at (2.5,0) (C) {$\begin{array}{c}C \\ 3 \end{array}$} ;
\node[state] at (2.5,-3) (E) {$\begin{array}{c}D \\ (-\Delta,0) \end{array}$} ;
\node[accepting,state] at (-0.5,-3) (T) {$\begin{array}{c}T \\ 0 \end{array}$} ;

\path (B) edge  node[above] {$x=1?$}node[below]{$x:=0$} (C);
\path (C) edge  node[left] {$x=1?$}node[right]{$x:=0$} (E);
\path (E) edge  node[above] {$x=1?$}node[below]{$x:=0$} (T);
  \end{tikzpicture}}
\caption{The gadget $D_{<}$}
\label{dec-c1-less}
\end{center}
\end{figure}

\item Zero Check Instruction.  $\ell_i$: if $c_2=0$ goto $\ell_1$ else goto $\ell_2$.

Figure \ref{zero-check} describes the gadget for zero check 
of counter $c_2$. No time is spent 
at location $\ell_i$, and the scheduler makes  a guess about the 
value of $c_2$. If he guesses that $c_2$ is zero, then he will choose 
the location $Z$. The environment can either allow the scheduler to go ahead 
with the simulation by choosing a rate 0, or could verify the correctness of scheduler's guess
by choosing a rate 100.  One unit of time has to be spent at the location $Z$. 
Thus, if the scheduler decides to verify 
and chooses the rate 100, the value of $z$ will be $105-old$. 
The environment will check if $old=\frac{1}{2^{c_1}}$, for some $c_1 \geq 0$. 
If the environment chooses 100, the best strategy for the scheduler 
is to choose the gadget $Z?$. Going to $\ell_1$ does not help the scheduler to win, since 
the environment can pick the rate 100 in all future choice locations, ensuring that 
the scheduler cannot win. 

\begin{figure}[h]
\begin{center}
\scalebox{0.7}{
\begin{tikzpicture}[->,>=stealth',shorten >=1pt,auto,node distance=1.8cm,
  semithick]
  \tikzstyle{every state}=[minimum size=3em,rounded rectangle]
    \node[initial,initial text={},state] at (-2,0)(A) {$\begin{array}{c}\ell_i \\ 0 \end{array}$} ;
   \node[state] at (0,2) (B) {$\begin{array}{c}NZ \\ \{100,0\} \end{array}$} ;
\node[state] at (0,-2) (C) {$\begin{array}{c}Z \\ \{100,0\} \end{array}$} ;
   \node[state] at (3,3) (B1) {$\ell_2$} ;
   \node[state,diamond] at (3,1) (B2) {$NZ?$} ;

\node[state] at (3,-1) (C1) {$l_1$} ;
   \node[state,diamond] at (3,-3) (C2) {$Z?$} ;

      \path (A) edge node[left] {$x=0?$} (B);
      \path (A) edge node[left] {$x=0?$} (C);
      
      \path (B) edge node[above] {$x=1?$} node[below]{$x,y:=0$}(B1);
      \path (B) edge node[above] {$x=1?$} node[below]{$x:=0$}(B2);

      \path (C) edge node[above] {$x=1?$} node[below]{$x,y:=0$}(C1);
      \path (C) edge node[above] {$x=1?$} node[below]{$x:=0$}(C2);

                
  
  \end{tikzpicture}}
\caption{Zero Check }
\label{zero-check}
\end{center}
\end{figure}
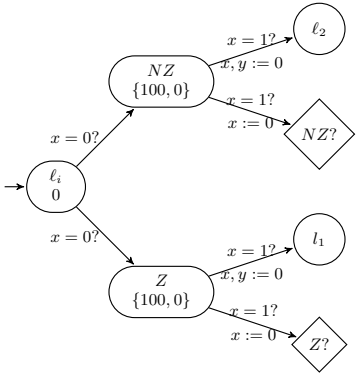

 The gadget $Z?$ given in figure \ref{z} is a check gadget which checks if $old= \frac{1}{2^{c_1}}$, for some $c_1 \geq 0$.  The gadget $Z?$ is entered with $z=105-old$, $x=y=0$. A time of unit is spent at location $L$, obtaining $z=5-old$
 and $x,y=0$. If indeed $c_2=0$, and if in addition, $c_1=0$, then $old=1$ and $z=4$. In this case, 
 the scheduler can go to the location $F$, spend a unit of time at $F$ obtaining $z=7$. 
 This leads to the location $G$, where the environment can pick any rate in $(0,\Delta)$. One unit of time 
 is spent in $G$, and in this case, we reach the mode $T$ with $z=7+\kappa \in B_{\Delta}(7)$, for 
 $\kappa \in (0,\Delta)$. Clearly, the scheduler wins here since his guess about $c_2$ being zero was correct.

In case $old=\frac{1}{2^{c_1}}$ for $c_1 >0$, then from location $M$, 
the scheduler cannot win by choosing location $F$ as the next location, 
since the value of $z$ on entry into $G$ will be 
$8-old$, where $old=\frac{1}{2^{c_1}}$ for $c_1 >0$. 
If $8-old > 7+\Delta$, then the scheduler has already lost. If 
$8-old \leq 7+\Delta$, let $8-old=(7+\Delta)-p$, for some $p \geq 0$.
Then $p=\Delta+old-1<\Delta$, since $old \leq \frac{1}{2}$.
Thus, the environment can pick a rate $p+\zeta \in (0, \Delta)$ 
such that $z=8-old+p+\zeta > 7+\Delta$.

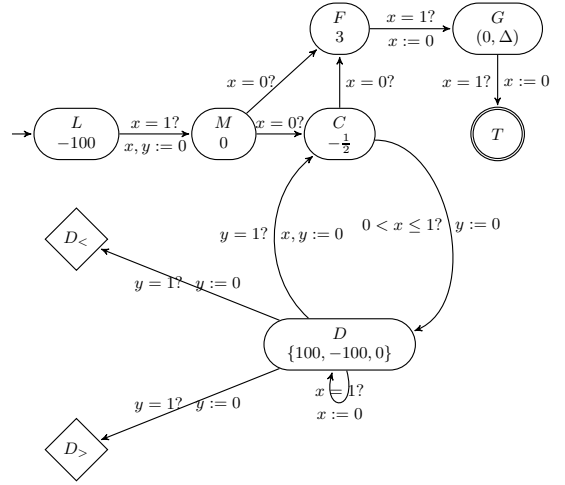
\begin{figure}[h]
\begin{center}
\scalebox{0.7}{
\begin{tikzpicture}[->,>=stealth',shorten >=1pt,auto,node distance=1.8cm,
  semithick]
  \tikzstyle{every state}=[minimum size=3em,rounded rectangle]
    \node[initial,initial text={},state] at (-2,0)(A) {$\begin{array}{c}L \\ -100 \end{array}$} ;
                 \node[state] at (0.8,0) (B) {$\begin{array}{c}M\\ 0 \end{array}$} ;
\node[state] at (3,0) (C) {$\begin{array}{c}C \\ -\frac{1}{2}\end{array}$} ;
\node[state] at (3,2) (F) {$\begin{array}{c}F \\ 3\end{array}$} ;
\node[state] at (6,2) (G) {$\begin{array}{c}G \\ (0,\Delta)\end{array}$} ;
\node[state,accepting] at (6,0) (T) {$T$} ;

      \node[state] at (3,-4) (D) {$\begin{array}{c}D \\ \{100, -100, 0\}\end{array}$} ;
     \node[state,diamond] at (-2,-2) (E1) {$D_{<}$} ;
      \node[state,diamond] at (-2,-6) (E2) {$D_{>}$} ;

      \path (A) edge node[above] {$x=1?$} node[below]{$x,y:=0$} (B);
            \path (B) edge node[above] {$x=0?$} (C);
            \path (B) edge node[left] {$x=0?$} (F);
            \path (C) edge node[right] {$x=0?$} (F);
      \path (F) edge node[above] {$x=1?$} node[below]{$x:=0$}(G);
      \path (G) edge node[left] {$x=1?$} node[right]{$x:=0$}(T);

      \path (C) edge[bend left=80] node[left] {$0< x \leq 1?$} node[right]{$y:=0$}(D);
      \path (D) edge node[left] {$y=1?$} node[right]{$y:=0$}(E1);
      \path (D) edge node[left] {$y=1?$} node[right]{$y:=0$}(E2);
      \path (D) edge [loop below] node[above]{$x=1?$} node[below]{$x:=0$}(D);
      \path (D) edge[bend left=50]  node[left]{$y=1?$} node[right]{$x,y:=0$}(C);
  \end{tikzpicture}}
\caption{The gadget $Z?$}
\label{z}
\end{center}
\end{figure}

Thus, if $c_1 > 0$, the best strategy for scheduler is to goto
location $C$. The subgraph consisting of locations $C,D$ and gadgets $D_{<}$ and 
$D_{>}$ simulates the decrement $c_1$ instruction. The ideal time to be spent 
at $C$ is $2old$ so that the value of $c_1$ is decremented by one. At location $D$,
the environment can choose a rate 0 (in which case, scheduler will go back to location $C$)
or a rate 100 (in which case scheduler will go to $D_{>}$) or a rate -100 (in which case, 
scheduler will go to $D_{<}$). In the case scheduler goes back to $C$, 
the new value of $z$ is $5-2old,x=y=0$. The ideal time to be spent at $C$ now is 
$4old$, and so on. At some point of time when $c_1=0$, we will obtain $z=4$. 
At this point, the scheduler can take the transition to $F$ from $C$, and as
seen above can reach $T$ with $z \in B_{\Delta}(7)$. If the scheduler 
goes to $F$ from $C$ when $z=5-old$ for some $0<old <1$, then 
as seen above, on entry into $G$, $z=8-old$, and the environment has a 
choice of rate in $(0, \Delta)$ such that scheduler loses.

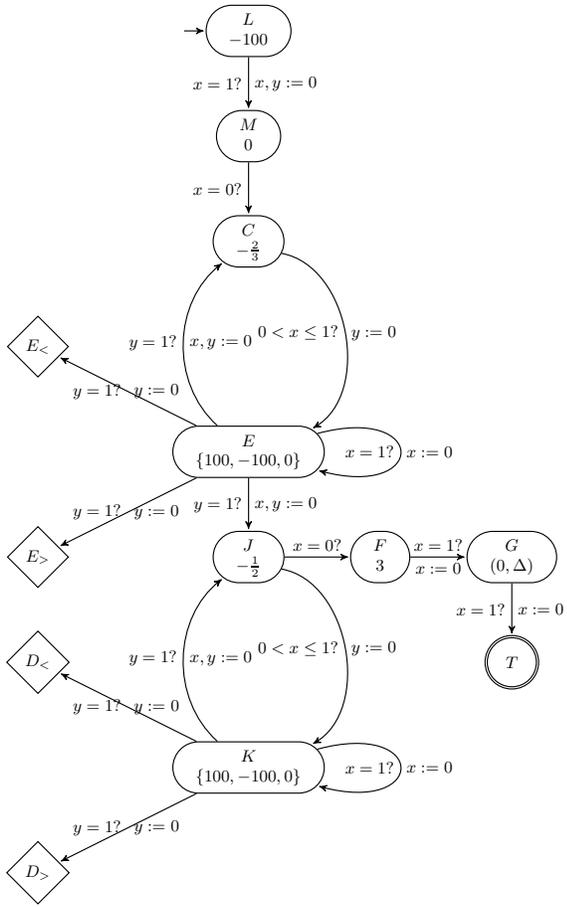
\begin{figure}[h]
\begin{center}
\scalebox{0.7}{
\begin{tikzpicture}[->,>=stealth',shorten >=1pt,auto,node distance=1.8cm,
  semithick]
  \tikzstyle{every state}=[minimum size=3em,rounded rectangle]
    \node[initial,initial text={},state] at (-4,2)(A) {$\begin{array}{c}L \\ -100 \end{array}$} ;
                 \node[state] at (-4,0) (B) {$\begin{array}{c}M\\ 0 \end{array}$} ;
\node[state] at (-4,-2) (C) {$\begin{array}{c}C \\ -\frac{2}{3}\end{array}$} ;
\node[state] at (-1.5,-8) (F) {$\begin{array}{c}F \\ 3\end{array}$} ;
\node[state] at (1,-8) (G) {$\begin{array}{c}G \\ (0,\Delta)\end{array}$} ;
\node[state,accepting] at (1,-10) (T) {$T$} ;
      \node[state] at (-4,-6) (D) {$\begin{array}{c}E \\ \{100, -100, 0\}\end{array}$} ;
      \node[state] at (-4,-8) (J) {$\begin{array}{c}J\\ -\frac{1}{2}\end{array}$} ;
      \node[state] at (-4,-12) (K) {$\begin{array}{c}K \\ \{100, -100, 0\}\end{array}$}; 
\node[state,diamond] at (-8,-10) (D1) {$D_{<}$} ;
      \node[state,diamond] at (-8,-14) (D2) {$D_{>}$} ;

\node[state,diamond] at (-8,-4) (E1) {$E_{<}$} ;
      \node[state,diamond] at (-8,-8) (E2) {$E_{>}$} ;
      \path (A) edge node[left] {$x=1?$} node[right]{$x,y:=0$} (B);
            \path (B) edge node[left] {$x=0?$} (C);
           \path (J) edge node[above] {$x=0?$} (F);
      \path (F) edge node[above] {$x=1?$} node[below]{$x:=0$}(G);
      \path (G) edge node[left] {$x=1?$} node[right]{$x:=0$}(T);

      \path (C) edge[bend left=70] node[left] {$0< x \leq 1?$} node[right]{$y:=0$}(D);
      \path (D) edge node[left] {$y=1?$} node[right]{$y:=0$}(E1);
      \path (D) edge node[left] {$y=1?$} node[right]{$y:=0$}(E2);
      \path (D) edge [loop right] node[left]{$x=1?$} node[right]{$x:=0$}(D);
      \path (K) edge [loop right] node[left]{$x=1?$} node[right]{$x:=0$}(K);
       \path (D) edge[bend left=50]  node[left]{$y=1?$} node[right]{$x,y:=0$}(C);
      \path (D) edge  node[left]{$y=1?$} node[right]{$x,y:=0$}(J);
      \path (J) edge[bend left=70] node[left] {$0< x \leq 1?$} node[right]{$y:=0$}(K);
      \path (K) edge[bend left=50]  node[left]{$y=1?$} node[right]{$x,y:=0$}(J);
 \path (K) edge node[left] {$y=1?$} node[right]{$y:=0$}(D1);
      \path (K) edge node[left] {$y=1?$} node[right]{$y:=0$}(D2);
    
   \end{tikzpicture}}
\caption{The gadget $NZ?$}
\label{nz}
\end{center}
\end{figure}

The gadget $NZ?$ is given in Figure \ref{nz}. This is entered into when the environment 
chooses a rate of  100 at location $NZ$ in Figure \ref{zero-check}. The idea is to verify that 
indeed $c_2$ is non-zero. Scheduler has to go through the locations 
$C,E$ atleast once so that $c_2$ is decremented atleast once (hence, $c_2 \neq 0$). 
The time elapse at $C$ must be $3old$, so that $z=5-old-\frac{2}{3}(3old)=5-3old$, decrementing $c_2$.
The gadgets $E_{<}$ and $E_{>}$ can be designed similar to 
the gadgets $D_{<}$ and $D_{>}$ to catch the errors of the scheduler when the time elapse 
is $3old+\epsilon$ and $3old-\epsilon$, $\epsilon >0$.  The scheduler must visit the $C,E$ loop
$c_2$ times (provided the environment gives rate 0 at location $E$ everytime). 
When the rate 0 is given at location $E$, 
scheduler can move to location $J$ when $c_2$ becomes 0. 
If $c_1$ is zero, then we get $z=4$ at the end of the $C,E$ loop. Then from $J$, scheduler can go to location $F$ 
spending no time at $J$, and reach $T$ with $z \in B_{\Delta}(7)$. However, if $c_1 >0$, 
then scheduler visits the $J,K$ loop until $c_1=0$ (provided the environment gives a rate 0 at location $K$).
When $c_1=0$, the scheduler can move from $J$ to $F$, and reach the target with $z \in B_{\Delta}(7)$. 
 \item The Halt location : The location labeled $Halt$ has rate 1. The scheduler will reach here 
 iff the two counter machine halts, and when the scheduler has simulated all the instructions correctly. The value of $z$ will be $5-old$, where $old=\frac{1}{2^{c_1}3^{c_2}}$, for $c_1, c_2 \geq 0$. A non-deterministic amount of time 
 can be spent by the scheduler here so that $z$ will lie in $B_{\Delta}(7)$. 
  \end{enumerate}    
It can be proved that the scheduler has a winning strategy to reach  $z \in B_{\Delta}(7))$ iff 
the two counter machine halts. 

\end{document}